\documentclass[11pt]{article}
\usepackage{amsmath,amsfonts,euscript,amssymb,amsthm,graphicx,verbatim, esint}
\usepackage{graphicx}
\usepackage{color}
\usepackage{hyperref}
\usepackage{enumerate}

\usepackage{mathtools}
\usepackage{caption}
\usepackage{subcaption}
\usepackage{authblk}

\newtheorem{theorem}{Theorem}
\newtheorem{corollary}{Corollary}
\newtheorem{lemma}{Lemma}
\newtheorem{proposition}{Proposition}
\newtheorem{remark}{Remark}

\newcommand{\R}{\mathbb{R}}
\newcommand{\N}{\mathbb{N}}
\newcommand{\Z}{\mathbb{Z}}

\newcommand{\C}{\mathbb{C}}

\newcommand*\lap{\mathop{}\!\mathbin\bigtriangleup}

\newcommand{\ein}{\boldsymbol{\hat{e}}}
\newcommand{\bs}[1]{ \boldsymbol{#1}}
\newcommand{\m}{\boldsymbol{m}}
\newcommand{\x}{\boldsymbol{x}}
\newcommand{\dif}{\ensuremath{\,\mathrm{d}}}
\newcommand{\grad}{\nabla}

\newcommand{\del}{\partial}
\newcommand{\tr}{\mathop{\mathrm{tr}}}

\renewcommand{\Re}{\mathrm{Re}\,}
\renewcommand{\Im}{\mathrm{Im}\,}

\DeclarePairedDelimiter{\abs}{\lvert}{\rvert} 
\DeclarePairedDelimiter{\norm}{\lVert}{\rVert} 
\DeclarePairedDelimiter{\bra}{(}{)} 
\DeclarePairedDelimiter{\pra}{[}{]} 
\DeclarePairedDelimiter{\set}{\{}{\}} 
\DeclarePairedDelimiter{\scp}{\langle}{\rangle} 

\AtEndDocument{%
 \par
 \medskip
 \begin{tabular}{@{}l@{}}%
 Department of Mathematics, RWTH Aachen University, Aachen Germany\\
 \textit{Email:} \texttt{xli@math1.rwth-aachen.de}\\
 \\
 Department of Mathematics and JARA Fundamentals of Future Information Technology,\\
 RWTH Aachen University, Aachen Germany \\
 \textit{Email:} \texttt{melcher@rwth-aachen.de}\\
 \end{tabular}}
 
\title{Lattice solutions in a Ginzburg-Landau model \\
for a chiral magnet}

\author{Xinye Li \qquad Christof Melcher}
\begin{document}
\maketitle

\begin{abstract}
We examine micromagnetic pattern formation in chiral magnets, driven by the competition of Heisenberg exchange, Dzyaloshinskii-Moriya interaction, easy-plane anisotropy and thermodynamic Landau potentials. Based on equivariant bifurcation theory we prove existence of lattice solutions branching off the zero magnetization state and investigate their stability. We observe in particular the stabilization of quadratic vortex-antivortex lattice configurations and instability of hexagonal skyrmion lattice configurations, and we illustrate our findings by numerical studies. 
\end{abstract}



\section{Introduction and main results}
Dzyaloshinskii-Moriya interaction (DMI) is the antisymmetric counterpart of Heisenberg exchange. It arises from the lack of inversion symmetry in certain magnetic system,
induced by the underlying crystal structures or by the system geometry in the presence of interfaces. In mircromagnetic models DMI arises in form of linear combinations of the so-called Lifshitz invariants, i.e. the components of the chirality tensor $\nabla \m \times \m$, and is therefore sensitive with respect to reflections and independent rotation in the domain and the target space, respectively. It is well-known that DMI gives rise to modulated phases. The basic phenomenon is that the energy of the homogeneous magnetization state $\m=const.$ can be lowered by means of spiralization in form of periodic domain wall arrays, the helical phase.
A prominent form of doubly periodic lattice state arises from the stabilization of topological structures, so-called chiral skyrmions, in two-dimensional chiral ferromagnets. 
Chiral skyrmions are localized structures which are topologically characterized by a unit $\mathbb{S}^2$ degree 
and a well-defined helicity depending on the specific form of DMI. In the presence of sufficiently strong perpendicular anisotropy and/or Zeeman field interaction, 
chiral skyrmions occur as local energy minimiziers in form of isolated topological solitons. Zeeman field interaction enables the possibility of an intermediate regime
where chiral skyrmion embedded into a hexagonal lattice are expected to be globally energy minimizing. Micromagnetic theories including DMI have been proposed in \cite{Bogdanov_Yablonskii:1989} with the idea that skyrmion lattice configurations represent a magnetic analogue of the mixed state in type-II superconductors. Corresponding phase diagrams and stability questions have been examined analytically and numerically in the seminal work \cite{Bogdanov_Hubert:1994, Bogdanov_Hubert:1999}, see \cite{Leonov:2016} for a recent review. A fully rigorous functional analytic theory on the existence, stability, asymptotics, internal structure and exact solvability of isolated chiral skyrmions has recently started to emerge \cite{Melcher:2014, Doering_Melcher:2017, Li_Melcher:2018, KMV_profile:2019, KMV_profile_large:2019, Barton_Ross_Schroers2018, Bernand_Mantel}. \\
Lattice states in two-dimensional Ginzburg-Landau models for chiral magnets including thermodynamic effects have been proposed theoretically in \cite{rossler2006spontaneous}, see also \cite{Muhlbauer_et_al:2009, yu2018transformation}.
Mathematically one may expect a close analogy with Abrikosov's vortex lattice solutions in Ginzburg-Landau models for superconductors \cite{Abrikosov} with a well-established theory in mathematical analysis. The occurrence of Abrikosov lattices in the framework of gauge-periodic solutions of Ginzburg-Landau equations and the optimality of hexagonal lattices have been thoroughly investigated by means of variational methods and bifurcation theory \cite{barany1992bifurcations, odeh1967existence, tzaneteas2009abrikosov, aydi_sandier, aftalion_blanc_nier, sandier_serfaty_CMP}. 

Ginzburg-Landau models in micromagnetics are, in contrast to superconductivity, directly formulated in terms of the physically observable magnetization field. A class of such models has been proposed and examined computationally in \cite{rossler2006spontaneous, Muhlbauer_et_al:2009}. Compared to the purely ferromagnetic case $|\m|= const.$, Ginzburg-Landau models offer a larger variety of patterns, including vortex and half-skyrmion arrays of opposite in-plane winding on square lattices, and skyrmions on hexagonal lattices, see Figure \ref{fig:lattice}. 

\begin{figure}[htbp]
 \centering
 \begin{subfigure}{0.49\textwidth}
 \centering
 \includegraphics[width=\linewidth]{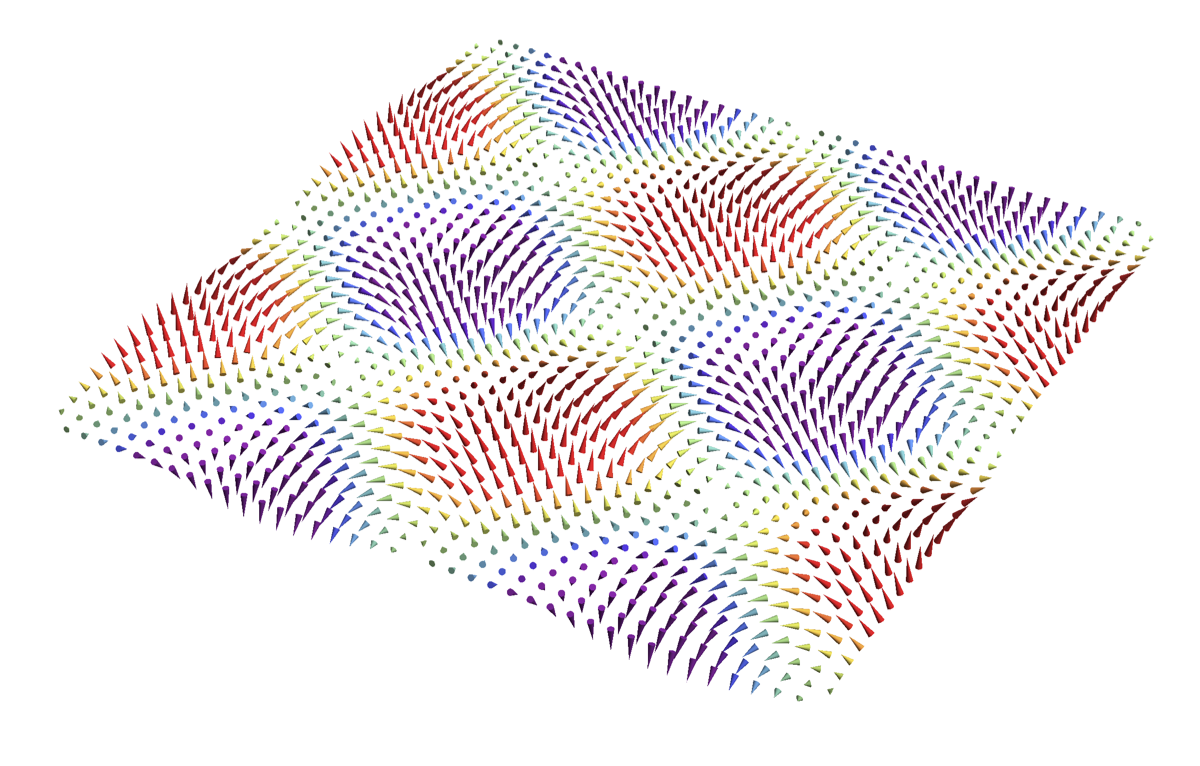}
 \caption{}
\end{subfigure}%
\hfill
\begin{subfigure}{0.49\textwidth}
 \centering
 \includegraphics[width=\linewidth]{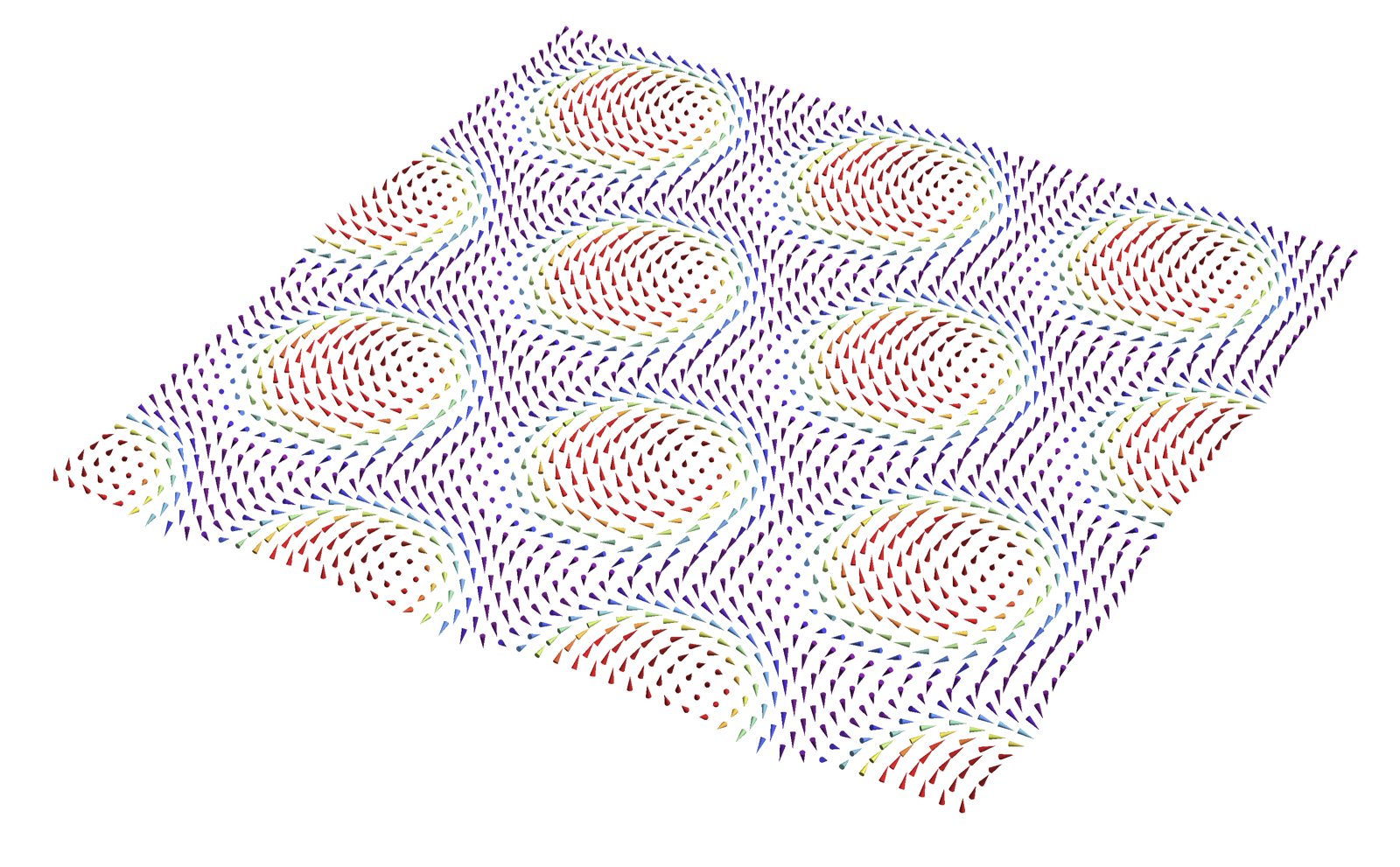}
 \caption{}
\end{subfigure}
\caption{The two-dimensional modulations on (a) a square lattice and (b) a hexagonal lattice.}
\label{fig:lattice}
\end{figure}
We shall examine the occurrence of periodic solutions near the paramagnetic state of zero magnetization, i.e., in a high temperature regime.
Our starting point are magnetization fields $\m:\R^2 \to \R^3$ governed by energy densities of the form
\[
A \abs{\grad \m}^2 + D \, \m \cdot (\grad \times \m) + f(\m) + K(\m \cdot \ein_3)^2.
\]
The Dirichlet term with $A>0$ is referred to as Heisenberg exchange interaction. The helicity term with $D \not=0$ is a prototypical form of DMI arising in the context of noncentrosymmetric cubic crystals. 
Note that for $\m=(m,m_3)$ defined on a two-dimensional domain 
\[
\nabla \times \m = \begin{pmatrix}
-\nabla^\perp m_3 \\ \nabla^\perp \cdot m
\end{pmatrix} \quad \text{where}\quad \nabla^\perp = (-\partial_2, \partial_1).
\]
A key analytical feature induced by DMI is the loss of independent rotational symmetry in magnetization space $\R^3$ and the domain $\R^2$. Finally, the Landau term $f$ is an even polynomial in the modulus $|\m|$
\[
a(T-T_C)\abs{\m}^2 + b\abs{\m}^4 + c \abs{\m}^6+d \abs{\m}^8 + \cdots
\]
where $T-T_C$ is the deviation from the Curie temperature. 
We shall focus on a minimal model with 
\[
a(T-T_C)\abs{\m}^2 + b \abs{\m}^4 \quad \text{where} \quad a,b>0. 
\]
For stability reasons we also include the easy-plane anisotropy $K(\m \cdot \ein_3)^2$ with $K>0$, which typically emerges as a reduced form of magnetostatic stray-field interaction in thin-film geometries, see e.g. \cite{Gioia_James}.

We are interested in magnetization fields $\m$ which are periodic with respect to a two-dimensional lattice 
$r\Lambda$ ($r>0$) with
\[
\Lambda=\frac{2\pi}{\Im \tau}(\Z \oplus \tau\Z)
\]

where 
$\tau$ is a complex number in the fundamental domain of the modular group, referred to as the lattice shape parameter, see Section \ref{sec:lattice}. 
Rescaling space we may assume $r=1$.
The  rescaled energy density reads
\begin{equation}\label{eq:density}
e(\m)= \frac{1}{2}\abs{\grad \m}^2 + \kappa \, \m \cdot (\grad \times \m) + \frac{\lambda}{2}\abs{\m}^2 + \frac{\alpha}{4} \abs{\m}^4+\frac{\beta}{2}(\m \cdot \ein_3)^2
\end{equation}
with dimensionless constants 
\begin{equation} \label{eq:scaling}
\kappa =\frac{Dr}{2A}, \quad \lambda=\frac{a(T-T_C)r^2}{A}, \quad \alpha=\frac{2br^2}{A} \quad \text{and} \quad \beta=\frac{K r^2}{A}.
\end{equation}
Since a sign reversal of the DM density can be achieved by reflections such as $m_3 \mapsto -m_3$, we may assume w.l.o.g that $\kappa>0$.

\paragraph{Euclidean Symmetry.}
The planar model to be examined arises from dimensional reduction.  
It is instructive to return to the original setting and consider the energy density on fields $\m$ from $\R^3$. In the case $\beta=\kappa=0$ we have
invariance with respect to the following action of the euclidean group in $\R^3$
\[
\m(\x) \mapsto {\bs R}\m(\bs{R}^{T}(\x-{\bs t}))
\quad \text{for} \quad \x \in \R^3 \quad \text{and} \quad (\bs{R},{\bs t})  \in O(3) \ltimes \R^3.
\]
In the case $\beta=0$ but $\kappa \not=0$ the reflection symmetry is broken and invariance is restricting to the special euclidean 
group with $\bs{R} \in SO(3)$, see Lemma \ref{lemma:curl}.
Including anisotropy $\beta \not=0$ and $\kappa \not= 0$ amounts to a further restriction of the rotation group to
elements of the form
\begin{equation} \label{eq:group_embedding}
\bs{R} = \begin{pmatrix} R & 0 \\ 0 & \det  R \end{pmatrix} \quad \text{where} \quad R \in O(2)
\end{equation}
defining an embedding $O(2) \hookrightarrow SO(3)$. Restricting to only horizontal translations $\bs t=(t,0)$ with $t \in \R^2$ amounts to
invariance of the two-dimensional model with respect to the action of the euclidean group in $\R^2$
\begin{equation} \label{eq:euclidean}
\m(x) \mapsto \bs{R} \m( R^{T}(x-t))
\quad \text{for} \quad x \in \R^2 \quad \text{and} \quad (R,t)  \in O(2) \ltimes \R^2
\end{equation}
where $\bs{R} \in SO(3)$ is given by \eqref{eq:group_embedding}.

We shall investigate the occurrence and stability of non-trivial $\Lambda$-periodic critical points $\m$ of the 
average energy over a primitive cell $\Omega_\Lambda$
\begin{equation}\label{eq:averageEnergy}
E_\Lambda (\m):= \fint_{\Omega_\Lambda} e(\m) \dif x,
\end{equation}
i.e. of non-trivial $\Lambda$-periodic solutions $\m$ to the Euler-Lagrange equation 
\begin{equation}\label{eq:main}
F(\m,\lambda)=-\lap \m + 2 \kappa \grad \times \m + \lambda\m + \alpha\abs{\m}^2 \m + \beta m_3 \ein_3= 0.
\end{equation}
We first discuss energy minimizing solutions.
\begin{theorem}\label{thm:minimizing}
Suppose $\alpha>0$, $\beta\ge 0$, and $\kappa>0$. 
\begin{enumerate}[(i)]
\item If $\lambda > \kappa^2$, then $\m\equiv 0$ is the unique energy minimizer on every lattice.
\item If $\lambda < \kappa^2$ and $\beta=0$, then the helix
\begin{equation} \label{eq:helix}
\m(x) = M \left(0, \cos(\kappa x_1), \sin(\kappa x_1)\right)
\quad \text{where} \quad M=\sqrt{\frac{\kappa^2-\lambda}{\alpha}}
\end{equation}
 (see Figure \ref{fig:helix}) is up to a joint rotation the unique energy minimizer on suitable lattices. 
\end{enumerate}
\end{theorem}

In the isotropic case $\beta=0$, Theorem \ref{thm:minimizing} indicates the existence of only two phases, paramagnetic or helical, while the picture in the 
anisotropic case $\beta>0$ is incomplete. The occurrence of helical phases is common to other mathematically related theories for condensed matter such as the Oseen-Frank 
model for chiral liquid crystals (see e.g. \cite{Virga_book}) or the Gross-Pitaevskii model for spin-orbit coupled Bose-Einstein condensates (see e.g. \cite{aftalion_rodiac}).
Helical structures in chiral ferromagnets are also discussed in \cite{Davoli_DiFratta2020, Muratov_Slastikov2016}.

Here we are interested in doubly periodic solutions. Given a lattice $\Lambda$, we aim to find $\lambda$ and a non-trivial $\Lambda$-periodic solution $\m$ 
of \eqref{eq:main} at $\lambda$. We call such pairs $(\m, \lambda)$ \emph{$\Lambda$-lattice solutions} and will prove the following:

\begin{theorem}\label{thm:main} 
Suppose $\alpha >0$, $\beta \ge 0$, $\kappa > 0$, and $\Lambda = \frac{2\pi}{\Im \tau}(\Z \oplus \tau\Z)$.
Then \eqref{eq:main} has a branch of $\Lambda$-lattice solution $(\m_s,\lambda_s)$, analytically parameterized by a real parameter $s$ near $0$,  in a neighbourhood of $\m_0 \equiv 0$ 
and 
\begin{equation}\label{eq:lambda0}
\lambda_0 = - 1 - \frac{\beta}{2}\pm \sqrt{4\kappa^2+\frac{\beta^2}{4}},
\end{equation}
provided $\lambda_0$ satisfies the non-resonances condition that 
\begin{equation}\label{eq:resonance}
\lambda_0 \neq - \abs{\omega}^2 - \frac{\beta}{2}\pm \sqrt{4\kappa^2\abs{\omega}^2+\frac{\beta^2}{4}}
\end{equation}
for any $\omega \in \Lambda^* \setminus \mathbb{S}^1$ where $\Lambda^*$ denotes the dual lattice, see Sec. \ref{sec:dual}.

The branch $(\m_s, \lambda_s)$ has the form
\begin{equation}\label{eq:BifurcationSolution}
\m_s = s\bs \varphi_1 + O(s^3), \quad
\lambda_s = \lambda_0 + s^2 \nu_2 + O(s^4)
\end{equation}
as $s \to 0$ with $\nu_2<0$ and $\bs \varphi_1$ explicitly determined. Furthermore, we have
\begin{equation}\label{eq:lambda2}
\nu_2= -\alpha \frac{\scp{\abs{\bs \varphi_1}^4}}{\scp{\abs{\bs \varphi_1}^2}^2} < 0
\end{equation}
and
\begin{equation}\label{eq:energy}
E_{\Lambda}(\m_s, \lambda_s)= \frac{s^4}{4}\bra*{-\alpha\frac{\scp{\abs{\bs \varphi_1}^4}}{\scp{\abs{\bs \varphi_1}^2}^2}}+O(s^6)
\end{equation}
where $\scp{\cdot}$ denote the average over a primitive cell $\Omega_\Lambda$.
\end{theorem} 
The morphology of bifurcation solutions is related to symmetry properties of the underlying lattice. 
Depending on this, the first-order bifurcation solution $\bs \varphi_1$, arising from the first critical wave number,
indicates a threefold pattern formation:
\begin{enumerate}
\item \textit{helical pattern}: 
exists on all lattices, the first-order bifurcation solution \eqref{eq:phi1} is given by a single helical mode (see Figure \ref{fig:helix});
\item \textit{vortex-antivortex pattern}: exists on equilateral lattices,
the first-order bifurcation solution \eqref{eq:phi2} is a superposition of two helices propagating in different directions, see Figure \ref{fig:lattice}(a);
\item \textit{skyrmionic pattern}: exists only on hexagonal lattices,
the first-order bifurcation solution \eqref{eq:phi3} is a superposition of three helices propagating in distinct directions, see Figure \ref{fig:lattice}(b).
\end{enumerate}

\begin{figure}[htbp]
 \centering
 \includegraphics[width=0.6\textwidth]{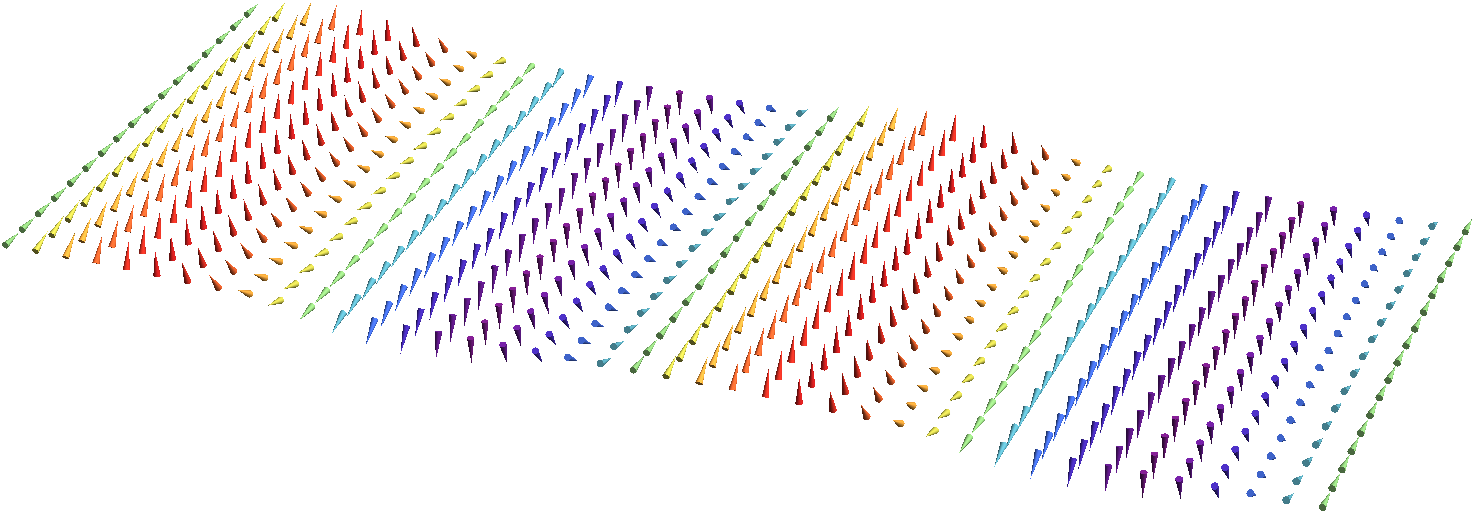} 
\caption{The one-dimensional modulations on a non-equilateral lattice.}
\label{fig:helix}
\end{figure}
We say that a bifurcation solution $(\m_s,\lambda_s)$ is \emph{linearly stable} under $\Lambda$-periodic perturbations if the linearized operator $L_s=D_{\m} F(\m_s,\lambda_s)$ is non-negative with a kernel that is only induced by translations, i.e.,
\[
\ker L_s = \mathrm{span}\, \{ \del_1 \m_s, \del_2 \m_s\}.
\]

Reducing the domain of the bifurcation parameter $s$ if necessary we shall prove the following stability properties
of quadratic vortex-antivortex and hexagonal skyrmion bifurcation solutions obtained in Theorem \ref{thm:main}:

\begin{theorem}\label{thm:stability} 
 Let $\alpha>0$, $\beta \ge 0$ and $\kappa>0$ satisfying 
\begin{equation}\label{eq:lambda0+}
\lambda_0 = - 1 - \frac{\beta}{2} + \sqrt{4\kappa^2+\frac{\beta^2}{4}}>0
\end{equation}
and
\begin{equation}\label{eq:stability_kappa}
4\kappa^2 \le \sqrt{4\kappa^2+\frac{\beta^2}{4}}+\sqrt{4\kappa^2 \gamma^2 +\frac{\beta^2}{4}}
\end{equation}
where $\gamma$ is the second critical wave number, i.e., 
\[
\gamma= \left\{ 
\begin{array}{ll}
 \abs{\tau} & \text{for } \abs{\tau}>1,\\
\sqrt{2-2\cos \theta}, & \text{for } \abs{\tau}=1, \, \frac{\pi}{3}< \theta \le \frac{\pi}{2},\\
\sqrt{3} & \text{for } \abs{\tau}=1, \, \theta = \frac{\pi}{3},
\end{array}
 \right.
\]
depending on the lattice shape $\tau=\abs{\tau}e^{i\theta}$.
\begin{enumerate}[(i)]
\item On square lattices, the vortex-antivortex lattice solution is linearly stable under $\Lambda$-periodic perturbations if $\beta>\frac{4}{\sqrt{3}}\kappa$ and unstable if $\beta<\frac{4}{\sqrt{3}}\kappa$.
\item On hexagonal lattices, the skyrmion lattice solutions are unstable under $\Lambda$-periodic perturbations for any $\beta \ge 0$.
\end{enumerate}
If $\lambda_0 \le 0$ bifurcation solutions are unstable, independently of $\beta \ge 0$, \eqref{eq:stability_kappa}, 
and the lattice shape.
\end{theorem} 

\begin{remark}
Helical bifurcation solutions on the square lattice have the same transition point: 
stability for $\beta<\frac{4}{\sqrt{3}} \kappa$ and instability for $\beta>\frac{4}{\sqrt{3}} \kappa$, see \cite{Li2020}.
\end{remark}

\begin{corollary}The quadratic vortex-antivortex lattice configuration exists and is stable if
\[
\beta>\frac{4}{\sqrt{3}}\kappa, \quad \beta \ge \sqrt{16\kappa^4-24\kappa^2+1} \quad\text{and}\quad \beta<4\kappa^2-1.
\]
The admissible set of $(\kappa,\beta)$ is not empty, see Figure \ref{fig:para_range}.
\end{corollary}

\begin{figure}[htbp]
 \centering
 \includegraphics[width=0.4\textwidth]{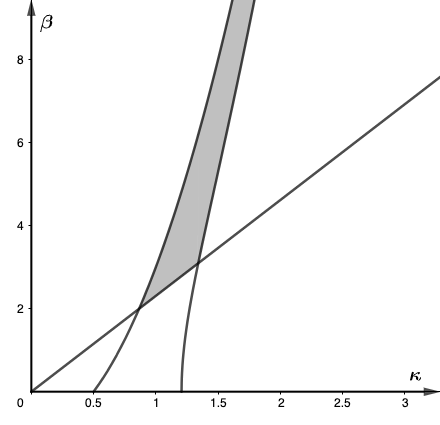} 
\caption{The admissible set of $(\kappa,\beta)$ for a stable quadratic vortex-antivortex lattice (indicated by the grey shaded area).}\label{fig:para_range}
\end{figure}

The existence and stability results are only an initial step towards understanding the stabilization of two-dimensional lattice solutions in chiral magnets.
In particular, stability of lattice solutions is only examined under the simplest perturbations which preserve lattice periodicity.
A more general stability result in the style of \cite{sigal_tzaneteas_2018} is beyond the scope of this work and requires a different approach.

The mathematical framework for our construction of lattice solutions is the equivariant branching lemma
\cite{Chossat_Lauterbach, golubitsky2012singularities}, a concept of symmetry-breaking bifurcation based on a particular type
of (axial) symmetry group. More precisely, letting
\begin{equation} \label{eq:Gamma_Lambda}
\Gamma_\Lambda=P_\Lambda \ltimes T_\Lambda,
\end{equation}
where $P_\Lambda \subset O(2)$ is the point group (or holohedry) of $\Lambda$ 
and $T_{\Lambda}=\R^2 / \Lambda$ is the torus of translations modulo $\Lambda$, the euclidean symmetry \eqref{eq:euclidean} induces an action of $\Gamma_\Lambda$ on spaces of $\Lambda$-periodic fields $\m$.  For each lattice we identify  all isotropy subgroups $\Sigma \subset \Gamma_\Lambda$ (up to conjugacy) so that the fixed subspace of $\Sigma$ in the kernel of linearized operator $D_{\m}F(\m_0,\lambda_0)$ is one-dimensional. By means of an equivariant Lyapunov-Schmidt procedure, \eqref{eq:main} reduces to a one-dimensional bifurcation equation. 
The implicit function theorem provides a solutions to the bifurcation equation in the one-dimensional fixed subspace of $\Sigma$, from which a solution to \eqref{eq:main} can be reconstructed. This solution is the bifurcation solution and inherits the symmetries featured by $\Sigma$.

In Section \ref{sec:preliminaries} we shall briefly recall the representation of lattices in the plane with an emphasis on symmetry
and Fourier series which are key to our bifurcation argument. In Section \ref{sec:ground_state} we shall derive energy bounds proving Theorem \ref{thm:minimizing}.
Solving the linearized version of Eq. \eqref{eq:main} explicitly by Fourier methods is the key ingredient to the proof of Theorem \eqref{thm:main} in Section \ref{sec:bifurcation} 
and provides insight about the morphology and topology of bifurcation solutions. In Section \ref{sec:stability} we investigate the stability of bifurcation solutions under $\Lambda$-periodic perturbations proving Theorem \ref{thm:stability} . 
Finally in Section \ref{sec:num} we validate our analytical results by a series of numerical simulations of gradient flows using a modified Crank-Nicolson scheme.
\section{Preliminaries}\label{sec:preliminaries}

\subsection{Representation of lattices}\label{sec:lattice}
Recall that a planar lattice $\Lambda$ is the integer span of two linearly independent vectors $t_1, t_2 \in \R^2$, i.e., 
\[
\Lambda= \{ m_1 t_1 + m_2 t_2: m \in \Z^2\}.
\]
Given $x \in \R^2$, a primitive cell of $\Lambda$ is a set of the form
\[
\Omega_{\Lambda} = \set*{x + a_1 t_1 + a_2 t_2, \, a_1,a_2 \in [0,1]}.
\]
The lattice basis $\{t_1, t_2\}$ is clearly non-unique. Identifying $\R^2 \cong \C$, however, the complex ratio $\tau \in \C$ of two basis vectors of $\Lambda$ 
contained in the fundamental domain
\begin{equation}\label{eq:tau}
\begin{aligned}
 T=\left\{ \abs{\tau} \ge 1: \Im \tau >0, -\frac{1}{2} < \Re \tau \le \frac{1}{2} 
\text{ and } \Re \tau \ge 0 \text{ if } \abs{\tau}=1\right \}
\end{aligned}
\end{equation}
parametrizes the lattice shape uniquely, see e.g. \cite{ahlfors1953complex}. Writing $\tau=\abs{\tau} e^{i \theta}$,
the range of $\abs{\tau}$ and $\theta$ corresponding to fundamental domain \eqref{eq:tau} is
\begin{equation}\label{eq:tau1}
\frac{\pi}{3} \le \theta < \frac{2\pi}{3} \text{ if } \abs{\tau} \ge 1
\quad\text{and}\quad
\frac{\pi}{3} \le \theta \le \frac{\pi}{2} \text{ if } \abs{\tau}=1.
\end{equation}
A lattice is called equilateral if $|\tau|=1$, where the borderline cases $\theta=\pi/2$ and $\theta = \pi/3$ are referred to as square and hexagonal lattice, respectively; 
other equilateral lattices are called rhombic.

There are two distinct types of symmetries preserving the lattice: the lattice translations and the holohedry group $P_{\Lambda}$, which is a finite subgroup of $O(2)$.  
Non-equilateral lattices have holohedry $\mathbb{Z}_2$ (oblique) or $D_2$ (rectangular); rhombic lattices have holohedry $D_2$; square lattices have holohedry $D_4$; hexagonal lattices have holohedry $D_6$, where $D_k$ is the dihedral group generated by rotation through $2\pi/k$ and a reflection, see e.g. \cite{Chossat_Lauterbach}.

\begin{figure}[htbp]
 \centering
 \includegraphics[width=0.6\textwidth]{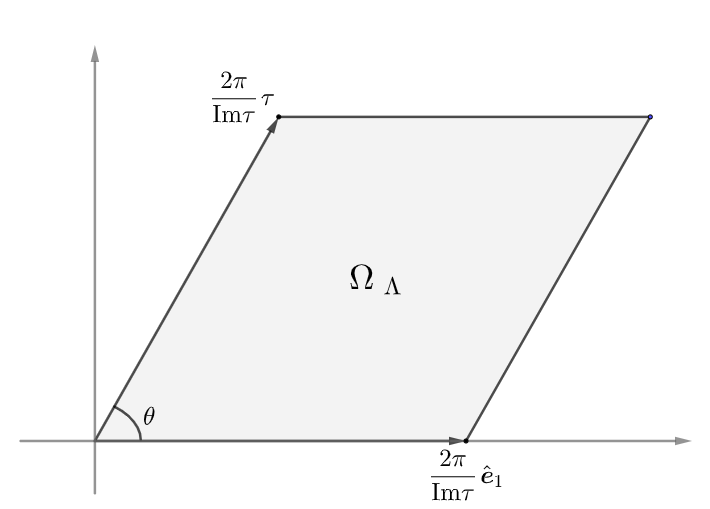}
\caption{The lattice cell $\Omega_{\Lambda}$ determined by lattice basis $\set{\frac{2\pi}{\Im \tau} \ein_1, \frac{2\pi}{\Im \tau} \tau}$}
\end{figure}

\subsection{Function spaces on lattices}
As the governing energy densities and the Euler-Lagrange equations are invariant under translation and joint rotation 
we can fix one basis vector of the lattice $\Lambda$ as $2\pi r \ein_1$ so that $\Lambda=2\pi r (\Z \oplus \tau \Z)$, is uniquely characterized by $r$ and $\tau$.
Upon rescaling we can arrange $\Lambda=\frac{2\pi}{\Im \tau} (\Z \oplus \tau \Z)$ spanned by $\set{\frac{2\pi}{\Im \tau} \ein_1, \frac{2\pi}{\Im \tau} \tau}$
and consider the rescaled density \eqref{eq:density} containing only dimensionless parameters.

For $\Lambda$-periodic functions or fields $f,g$ on $\R^2$ we denote the average by
\[
\scp{f} := \fint_{\Omega_\Lambda} f(x) \dif x = \frac{1}{\abs{\Omega_\Lambda}} \int_{\Omega_\Lambda} f(x) \dif x,
\]
the $L^2$ scalar product by
\[
\scp{f,g} := \fint_{\Omega_\Lambda} f(x) \cdot g(x) \dif x,
\]
and the $L^2$ norm by $\|f\|:= \sqrt{\scp{f, f}}$ once existent. Accordingly we define
\[
L^2_\Lambda :=\set{ f: \R^2 \to \R^3 \text{ $\Lambda$-periodic with } \|f\|< \infty}
\]
and for $k \in \N$ the Sobolev spaces
\[
H^k_\Lambda :=\set{f \in L^2_\Lambda: \del^{v}f \in L^2_\Lambda \text{ for all $|v|=k$}}
\]
which are subspaces of $L^2_{\rm loc}:=L^2_{\rm loc}(\R^2;\R^3)$ and $H^k_{\rm loc}:=H^k_{\rm loc}(\R^2;\R^3)$, respectively.
Thanks to Sobolev embedding the average energy \eqref{eq:averageEnergy}
defines an analytic functional on $H^1_{\Lambda}$. 
Critical points $\m \in H^2_{\Lambda}$ of $E_{\Lambda}$ satisfy 
\[
F(\m,\lambda)=0
\]
where $F: H^2_{\rm loc} \times \R \to L^2_{\rm loc}$ is the nonlinear operator given by
\begin{equation}\label{eq:main1}
F(\m,\lambda)=-\lap \m + 2 \kappa \grad \times \m + \lambda\m + \alpha\abs{\m}^2 \m +\beta m_3 \ein_3.
\end{equation}

\subsection{Dual lattice and Fourier series}\label{sec:dual}

Fourier expansion on $\Lambda$ requires the notion of dual lattice given by 
\[
\Lambda^* =\{v \in \R^2 : u \cdot v \in 2\pi \Z \text{ for all }u \in \Lambda\}.
\] 
In particular $\Lambda^*=\mathcal{A}^{-T}\Z^2$ for $\Lambda=2\pi \mathcal{A} \Z^2$ where in our setting
\begin{equation} \label{eq:lattice_matrix}
\mathcal{A}= \frac{1}{\abs{\tau}\sin \theta}\begin{pmatrix}1 & \abs{\tau}\cos \theta \\ 0 & \abs{\tau}\sin \theta \end{pmatrix}
\quad \text{and} \quad 
\mathcal{A}^{-T} = \begin{pmatrix} \abs{\tau}\sin \theta & 0 \\ -\abs{\tau}\cos \theta & 1 \end{pmatrix}.
\end{equation}

In the equilateral case $\abs{\tau}=1$, dual lattices remain square if $\theta=\pi/2$ and hexagonal if $\theta = \pi/3$. \\
For $f \in L^2_{\Lambda}$ and $v \in \Lambda^*$ Fourier coefficients 
are defined as
\[
\tilde{f}(v) = \fint_{\Omega_\Lambda} f(x)e^{- i \, v\cdot x} \dif x,
\]
and the following Fourier expansion
\[
f(x)=\sum_{v \in \Lambda^*} \tilde{f}(v) e^{ i \, v \cdot x}
\]
holds true in the $L^2$ sense along with Parseval's identity
\[
 \fint_{\Omega_\Lambda} \abs{f(x)}^2 \dif x = \sum_{v \in \Lambda^*} \abs{\tilde{f}(v)}^2.
\]

\subsection{Equivariance and lattice symmetry}\label{sec:symmetry}

The action of an element $\gamma =(R,t)$ of the group $\Gamma_\Lambda = P_\Lambda \ltimes T_\Lambda$, the semi-direct product of the holohedry of 
$\Lambda$ and translations modulo $\Lambda$, on a field $\m:\R^2 \to \R^3$ given by
\[
(\gamma \bullet \m)(x)= \bs{R} \m(R^{-1} (x-t)) \quad \text{for}\quad x \in \R^2
\]
where the corresponding $\bs{R} \in SO(3)$ is determined by \eqref{eq:group_embedding},
is an isometry on $H^k_\Lambda$ for every $k \in \N_0$, and the operator \eqref{eq:main1} is $\Gamma_\Lambda$-\emph{equivariant} in the sense that 
\[
F(\gamma \bullet \m, \lambda)=\gamma \bullet F(\m,\lambda)
\]
for all $\gamma \in \Gamma_\Lambda$, $\m \in H^2_\Lambda$ and $\lambda \in \R$, see Lemma \ref{lemma:curl}

The symmetry of a field 
$\bs \phi \in H^2_{\Lambda}$ is given in terms of the \emph{isotropy subgroup} 
\[
\Sigma_{\bs \phi} = \set{\sigma \in \Gamma_\Lambda: \sigma \bullet \bs \phi=\bs \phi},
\]
i.e., the largest subgroup of $\Gamma_\Lambda$ which fixes $\bs \phi$. Given a subspace $X \subset H^2_\Lambda$, the \emph{fixed subspace} associated to a subgroup $\Sigma \subseteq \Gamma_\Lambda$ in $X$ is 
\[
\mathrm{Fix}_X(\Sigma)=\set*{\bs \phi \in X: \; \sigma \bullet \bs \phi = \bs \phi}.
\]
$L^2_\Lambda$ orthogonal projections on such invariant subspaces of $H^2_\Lambda$ are equivariant.
Equivariance therefore propagates to the Lyapunov-Schmidt decomposition enabling a reduction of the bifurcation equation to $\mathrm{Fix}_{X_0}(\Sigma)$ where $X_0$ is the kernel of the linearization of $F$ at a bifurcation point $(0, \lambda_0)$, see e.g. \cite{Chossat_Lauterbach}. Bifurcation solutions arising from the equivariant branching lemma turn out to have full $\Sigma$ symmetry and are unique in this class. 

Anticipating the results in Sec. \ref{sec:bifurcation}, we introduce a set of isotropy subgroups $\Sigma_i \subseteq \Gamma_\Lambda$ on different lattice types, which play a central role in our bifurcation argument. In Proposition \ref{prop:fredholm}, we shall prove that the $\Sigma_i$ are indeed axial, i.e. have one-dimensional fixed-point subspace in the kernel $X_0$.

On non-equilateral lattices ($\abs{\tau} > 1$) we consider the symmetry group
\begin{equation}\label{eq:Sigma1}
\Sigma_1= \mathbb{Z}_2 \ltimes T_1
\end{equation}
where $T_1$ are the translations in $x_1$-direction and $\mathbb{Z}_2= \set{I, R}$ with associated $SO(3)$ elements
\[
\bs{I}=\begin{pmatrix}
1 & 0 & 0 \\ 0 & 1 & 0 \\ 0 & 0 & 1
\end{pmatrix}
\quad\text{and}\quad
\bs{R}=\begin{pmatrix}
-1 & 0 & 0 \\ 0 & -1 & 0 \\ 0 & 0 & 1
\end{pmatrix}.
\]
The corresponding bifurcation solutions feature a $\Sigma_1$-invariant pattern of helices propagating in the $x_2$ 
 direction (Figure \ref{fig:helix}).

Equilateral lattices ($\abs{\tau}=1$ and $\frac{\pi}{3} \le \theta \le \frac{\pi}{2}$) have an additional symmetry given by reflections across the diagonals of the lattice cell.
Therefore, in this case, we consider, in addition to $\Sigma_1$,  the symmetry group
\begin{equation}\label{eq:Sigma2}
\Sigma_2= \set{I, R, R_\theta^+, R_\theta^-} = D_2
\end{equation}
where the associated $SO(3)$ elements are
\[
\bs{R}_\theta^+=\begin{pmatrix}
\cos \theta & \sin \theta & 0 \\ \sin \theta & -\cos \theta & 0 \\ 0 & 0 & -1
\end{pmatrix}
\quad\text{and}\quad
\bs{R}_\theta^-=\begin{pmatrix}
-\cos \theta & -\sin \theta & 0 \\ -\sin \theta & \cos \theta & 0 \\ 0 & 0 & -1
\end{pmatrix}.
\]
Bifurcation solutions corresponding to $\Sigma_2$ are doubly periodic array of vortices and antivortices (Figure \ref{fig:lattice}(a)).
 
For the hexagonal lattice ($\abs{\tau}=1$ and $\theta=\frac{\pi}{3}$), in addition to $\Sigma_1$ and $\Sigma_2$, the symmetry group considered is the cyclic subgroup $\mathbb{Z}_6$ of $D_6$ generated by rotation through $\frac{\pi}{3}$
\begin{equation}\label{eq:Sigma3}
\Sigma_3 = \set{R_k, \, k=0,1,\cdots,5}
\end{equation}
where the corresponding $SO(3)$ elements are
\[
\bs{R}_k = \begin{pmatrix}
\cos \frac{k\pi}{3} & - \sin \frac{k\pi}{3} & 0 \\
\sin \frac{k\pi}{3} & \cos \frac{k\pi}{3} & 0 \\
0 & 0 & 1
\end{pmatrix}.
\]
Bifurcation solutions corresponding to $\Sigma_3$ are hexagonal skyrmion lattices (Figure \ref{fig:lattice}(b)).

\section{Energy bounds on lattices}\label{sec:ground_state}
For a lattice $\Lambda$ and $e(\m)$ given by \eqref{eq:density}, we examine ansatz-free lower bounds 
\begin{equation} \label{eq:energy_functional}
E_\Lambda(\m) = \fint_{\Omega_\Lambda} e(\m) \dif x.
\end{equation}
We start by expressing the total exchange density
\[
e_0(\m) = \frac 1 2 |\nabla \m|^2 + \kappa \, \m \cdot (\nabla \times \m)
\]
as a sum of sign definite terms and a null Lagrangian also know as Frank's formula in the theory
of liquid crystals, see e.g. \cite{Virga_book} Chapter 3.

\begin{lemma} \label{lemma:Frank}
For $\m \in H^1_{\rm loc}(\R^3; \R^3)$ the following holds
\[
e_0(\m) = \frac 1 2 
\bra*{ \bra*{\nabla \cdot \m}^2 + |\nabla \times \m + \kappa \m|^2 - \kappa^2 |\m|^2
+ \nabla \cdot \pra*{ (\m \cdot \nabla) \m - \m (\nabla \cdot \m) } }
\]
 in the sense of distributions. In particular for $\m \in H^1_{\Lambda}$
\[
E_\Lambda(\m) \ge \fint_{\Omega_\Lambda}
 \frac{\abs{\m}^2}{2} \bra*{\lambda -\kappa^2 +\frac{\alpha}{2} \abs{\m}^2} \, \dif x
\]
with equality if and only if $\nabla \times \m + \kappa \m=0$ and $\beta=0$.
\end{lemma}

From the lemma we obtain immediately claim (i) in Theorem \ref{thm:minimizing}, and moreover:

\begin{proposition} 
If $\lambda < \kappa^2$ the energy admits a lower bound
\[
E_\Lambda(\m) \ge 
- \frac{ (\lambda -\kappa^2)^2}{4 \alpha}
\]
which for $\beta=0$ is precisely attained for $\m$ of constant modulus 
\[
|\m|= \sqrt{ \frac{\kappa^2-\lambda}{\alpha}}
\]
and such that $\nabla \times \m + \kappa \m=0$.
\end{proposition}

The unimodular Beltrami fields being parallel to their curl have been classified by Ericksen within the variational theory of liquid crystals,
see e.g. \cite{Virga_book} and references therein. For the present case of constant $\kappa$ those are helices of pitch $2\pi/|\kappa|$, i.e., \eqref{eq:helix}.
For the convenience of the reader we present this fundamental result in the Appendix Lemma \ref{lemma:Beltrami}, which yields claim (ii) in Theorem \ref{thm:minimizing}. 
Thus in order to realize the lower energy bound, the underlying lattice $\Lambda$ is required to accommodate such a helix. In this case the zero state loses its linear stability 
at $\lambda=\kappa^2$. In fact, the Hessian 
\[ 
H_\Lambda(\m) \langle \bs \phi, \bs \phi \rangle = \frac{d^2}{ds^2}\Bigg|_{s=0} E_{\Lambda}(\m+s \bs \phi)
\]
for $\m, \bs \phi \in H^1_\Lambda$ at $\m \equiv 0$ reads
\[
H_\Lambda(0) \langle \bs \phi, \bs \phi \rangle = \fint_{\Omega_\Lambda} |\nabla \bs \phi|^2 + 2 \kappa \, \bs \phi \cdot (\nabla \times \bs \phi) + \lambda |\bs \phi|^2 \, \dif x.
 \]
By the preceding arguments it satisfies 
\[
 H_{\Lambda}(0) \langle \bs \phi, \bs \phi \rangle \ge (\lambda - \kappa^2) 
\scp{|\bs \phi|^2}
\]
and has a helical instability at $\lambda=\kappa^2$.

\section{Bifurcation on lattices}\label{sec:bifurcation}

In this section we prove Theorem \ref{thm:main} based on the equivariant branching lemma.
The requisite assumptions are summarized in Proposition \ref{prop:fredholm}. Bifurcation points $\lambda_0$ 
arising from the first critical wave number are identified by means of a Fourier expansion in Lemma \ref{la:lambda} 
in combination with Lemma \ref{la:nu}.

We examine the linearization
\[
L^{(\lambda)}=D_{\m} F(0,\lambda):H^2_\Lambda \to L^2_\Lambda
\] 
of $F$ at $\m=0$ for arbitrary $\lambda$ given by
\begin{align*}
L^{(\lambda)} \bs \phi 
&= -\lap \bs \phi + 2\kappa \grad \times \bs \phi +\lambda \bs \phi + \beta \phi_3 \ein_3.
\end{align*}
We need to find non-trivial $\Lambda$-periodic solutions $\bs \phi$ of the equation
\begin{equation}\label{eq:linear}
-\lap \bs \phi + 2\kappa \grad \times \bs \phi +\lambda \bs \phi + \beta \phi_3 \ein_3= 0
\end{equation}
for $\lambda=\lambda_0$ depending on $\kappa$ and $\beta$.
\begin{lemma}\label{la:lambda}
Eq. \eqref{eq:linear} admits non-constant solutions $\bs \phi \in H_\Lambda^2$ if and only if
\[ 
\lambda=-\abs{v}^2 - \frac{\beta}{2}\pm \sqrt{4\kappa^2\abs{v}^2+\frac{\beta^2}{4}}
\]
for some $v \in \Lambda^*\setminus \set{0}$. 
\end{lemma}
\begin{proof}
We expand $\bs \phi \in H_\Lambda^2$ in Fourier series
\[
\bs \phi(x)=\sum_{v \in \Lambda^*} \bs \phi_{v} e^{i v \cdot x} =\sum_{v \in \Lambda^*} \begin{pmatrix} a_{v} \\ b_{v} \\ c_{v}\end{pmatrix}
e^{iv \cdot x} 
\]
for $x \in \R^2$, where, recalling \eqref{eq:lattice_matrix},
\begin{equation}\label{eq:k_nu}
v = (v_1, v_2)=\mathcal{A}^{-T} k = 
\begin{pmatrix} \Im \tau k_1 \\ -\Re \tau k_1 + k_2\end{pmatrix}.
\end{equation}
Since
\[
-\lap \bs \phi (x) = \sum_{v \in \Lambda^*} \abs{v}^2\begin{pmatrix} a_{v} \\ b_{v} \\ c_{v}\end{pmatrix}
e^{i v \cdot x}
\quad\text{and}\quad
\grad \times \bs \phi = \sum_{v \in \Lambda^*}\begin{pmatrix} 
i v_2 c_{v} \\ - i v_1 c_{v} \\ i (v_1 b_{v} - v_2 a_{v})
\end{pmatrix} e^{ i v \cdot x},
\]
the linearized equation \eqref{eq:linear} is equivalent to the system
\[
\begin{aligned}
&\abs{v}^2 a_v + 2\kappa i v_2 c_v + \lambda a_v = 0\\
&\abs{v}^2 b_v - 2\kappa i v_1 c_v + \lambda b_v = 0 \\
&\abs{v}^2 c_v + 2\kappa i (v_1 b_v - v_2 a_v) + \lambda c_v + \beta c_v = 0,
\end{aligned}
\]
for all $v \in \Lambda^*$. Constant solutions with $v=0$ exist only if $\lambda=0$ or $\lambda+\beta=0$. For $v \neq 0$, we have
\begin{equation} \label{eq:F_solution}
a_{v} = -\frac{2\kappa i v_2}{\abs{v}^2 +\lambda} c_{v},\quad 
b_{v} = \frac{2\kappa i v_1}{\abs{v}^2 +\lambda} c_{v}
\end{equation}
and 
\[
 \bra*{\abs{v}^2 +\lambda+\beta - \frac{4\kappa^2 \abs{v}^2}{\abs{v}^2 +\lambda}} c_{v} = 0
\]
which is possible for $c_{v} \not=0$ only if $\displaystyle{\lambda=-\abs{v}^2 - \frac{\beta}{2}\pm \sqrt{4\kappa^2\abs{v}^2+\frac{\beta^2}{4}}}$.
\end{proof}

We focus on the wave vectors $v \in \Lambda^*\setminus \set{0}$ of shortest length which are characterized by 
 minimizing problem 
\[
\abs{v}^2 = \min_{k \in \Z^2 \setminus \{0\}} \bra*{(\Im \tau\, k_1)^2 + (\Re \tau\, k_1 - k_2)^2}.
\]
A straightforward analysis yields (see Figure \ref{fig:crit_nu}):
\begin{lemma}\label{la:nu}
Let $v \in \Lambda^*\setminus \set{0}$. Then $\abs{v}$ is minimized by
\begin{enumerate}[(i)]
\item $k=(0,\pm 1)$ if $\abs{\tau}>1$,
\item $k=(\pm 1,0), (0, \pm 1)$ if $\abs{\tau}=1$, $\frac{\pi}{3}<\theta \le \frac{\pi}{2}$, and
\item $k=(\pm 1,0), (0, \pm 1), \pm(1,1)$ if $\abs{\tau}=1$, $\theta=\frac{\pi}{3}$.
\end{enumerate}
\end{lemma}
Possible bifurcations at $\lambda$ corresponding to wave numbers larger than $\abs{v} = 1$ turn out to be unstable, see proof of Lemma \ref{la:L0} below.
Hence we shall consider bifurcation occurring at
\[
\lambda_0 = - 1 - \frac{\beta}{2}\pm \sqrt{4\kappa^2+\frac{\beta^2}{4}}
\]
satisfying \eqref{eq:resonance}, which guarantees that only the first non-trivial wave number contributes to the kernel of the linearization.
\begin{figure}[htbp]
 \centering
 \hfill
 \begin{subfigure}[b]{0.23\textwidth}
 \centering
 \includegraphics[width=\textwidth]{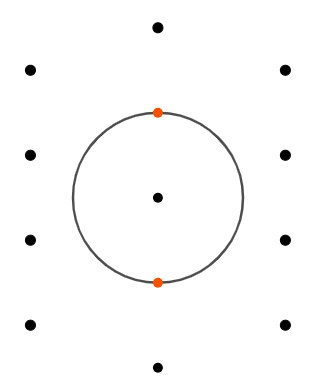}
 \caption{}
 \end{subfigure}
 \hfill
 \begin{subfigure}[b]{0.23\textwidth}
 \centering
 \includegraphics[width=\textwidth]{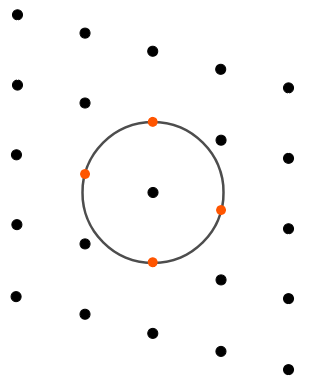}
 \caption{}
 \end{subfigure}
 \hfill
 \begin{subfigure}[b]{0.23\textwidth}
 \centering
 \includegraphics[width=\textwidth]{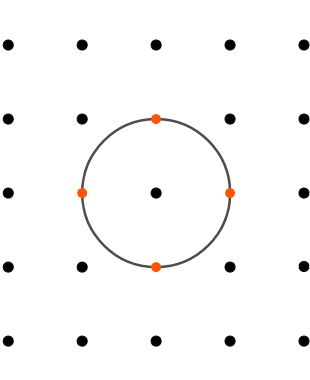}
 \caption{}
 \end{subfigure}
 \hfill
 \begin{subfigure}[b]{0.23\textwidth}
 \centering
 \includegraphics[width=\textwidth]{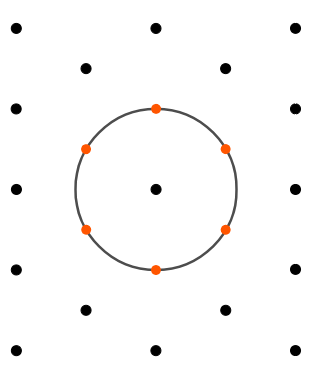}
 \caption{}
 \end{subfigure}
 \caption{The critical circle with radius 1 
 around 0 and the critical wave vectors (red points) on the dual lattice of a (a) non-equilateral lattice, (b) rhombic lattice, (c) square lattice and (d) hexagonal lattice.}
 \label{fig:crit_nu}
\end{figure}

\begin{proposition}\label{prop:fredholm}
Under the assumptions of Theorem \ref{thm:main}, the linearization
\[
L_0:= D_{\m}F(0,\lambda_0): \quad H^2_\Lambda \subset L^2_\Lambda \to L^2_\Lambda
\]
is a self-adjoint Fredholm operator with
$\dim \mathrm{ker}\, L_0 = N$, where $N=2$ on non-equilateral lattices, $N=4$ on rhombic and square lattices and $N=6$ on hexagonal lattices. The fixed subspace of 
the corresponding symmetry group $\Sigma_i$, $i=1,2,3$, (see section \ref{sec:symmetry}) in $X_0=\mathrm{ker}\, L_0$ is one-dimensional,
\[
 \mathrm{Fix}_{X_0}(\Sigma_i)= \mathrm{span}\,\set{\bs \varphi_1^{(i)}}.
\]
\end{proposition}
\begin{proof}
The normalized Fourier coefficients \eqref{eq:F_solution} obtained in the proof of Lemma \ref{la:lambda} are
\begin{equation}\label{eq:linearSol}
\bs \phi_{v} = \frac{1}{\sqrt{1+A^2}}\begin{pmatrix}
A\begin{pmatrix}
-\frac{i v_2}{\abs{v}} \\ \frac{i v_1}{\abs{v}} 
\end{pmatrix}
\\ 1
\end{pmatrix}
\quad\text{where}\quad
A=\frac{2\kappa}{-\frac{\beta}{2} \pm \sqrt{4\kappa^2+\frac{\beta^2}{4}}}.
\end{equation}
The corresponding real-space solutions of \eqref{eq:linear} to the wave vector $v$ are
\[
\bs \phi_{1,v}(x)= \bs \phi_{v} e^{i v \cdot x}+ \bs \phi_{-v}e^{-i v \cdot x}
= \frac{1}{\sqrt{1+A^2}}\begin{pmatrix}
A \begin{pmatrix}
\frac{v_2}{\abs{v}} \sin(v \cdot x) \\
-\frac{v_1}{\abs{v}} \sin(v \cdot x) 
\end{pmatrix}\\
\cos(v \cdot x)
\end{pmatrix}
\]
and
\[
\bs \phi_{2,v}(x)= i\bs \phi_{v}e^{i v \cdot x} + i\bs \phi_{-v}e^{-i v \cdot x}
= \frac{1}{\sqrt{1+A^2}}\begin{pmatrix}
A \begin{pmatrix}
\frac{v_2}{\abs{v}} \cos(v \cdot x) \\
-\frac{v_1}{\abs{v}} \cos(v \cdot x)
\end{pmatrix} \\
-\sin(v \cdot x)
\end{pmatrix}
\]
satisfying
\begin{equation} \label{eq:phi_sym}
\bs \phi_{1,-v}=\bs \phi_{1,v} \quad \text{and} \quad
\bs \phi_{2,-v}= -\bs \phi_{2,v}.
\end{equation}

We need to find all solutions of $L_0 \bs \phi =0$ for a given wave number $\abs{v}$. According to Lemma \ref{la:nu}, we consider following cases separately.

\underline{Case 1}: Non-equilateral lattice, $\abs{\tau}>1$

The wave vector \eqref{eq:k_nu} corresponding to $k=(0,1)$ is
\[
 v^{(1)} =  \begin{pmatrix} 0 \\ 1 \end{pmatrix} \in \Lambda^*.
\]
Therefore, taking into account \eqref{eq:phi_sym},
\[
\mathrm{ker}\,L_0=\mathrm{span } \set{\bs \phi_{1,v^{(1)}}, \bs \phi_{2,v^{(1)}}}.
\]
$\Sigma_1$ given in \eqref{eq:Sigma1} is the only axial isotropy subgroup. More precisely
\[ \mathrm{Fix}_{X_0}(\Sigma_1)=\mathrm{span } \set{\bs \varphi_1^{(1)}}\]
with the $L_\Lambda^2$ normalized
\begin{equation}\label{eq:phi1}
\bs \varphi_1^{(1)}=\sqrt{2}\bs \phi_{1,v^{(1)}}= \sqrt{\frac{2}{1+A^2}}\begin{pmatrix}
A\sin(x_2) \\
0 \\
\cos(x_2)
\end{pmatrix}.
\end{equation}

\underline{Case 2}: Rhombic or square lattice, $\abs{\tau}=1$, $\frac{\pi}{3} < \theta \le \frac{\pi}{2}$

The wave vectors \eqref{eq:k_nu} corresponding to $k = (1,0)$ and $(0, 1)$
\[
v^{(2)} = \begin{pmatrix} \sin \theta \\ -\cos \theta \end{pmatrix} \quad\text{and}\quad
v^{(3)} = \begin{pmatrix}0 \\ 1\end{pmatrix}.
\]
Therefore, taking into account \eqref{eq:phi_sym},
\[
\mathrm{ker}\,L_0=\mathrm{span } \set{\bs \phi_{i,v^{(j)}}, \; i=1,2, \, j=2,3}.
\]
As in Case 1 we have 
\[
\mathrm{Fix}_{X_0}(\Sigma_1)=\mathrm{span } \set{\bs \varphi_1^{(1)}} \quad \text{where} \quad \bs \varphi_1^{(1)}=\sqrt{2}\bs \phi_{1,v^{(2)}}.
\]
The fixed subspace of $\Sigma_2$ (see \eqref{eq:Sigma2}) in the kernel $X_0$ is
\[ \mathrm{Fix}_{X_0}(\Sigma_2)=\mathrm{span } \set{\bs \varphi_1^{(2)}}\] 
with the $L^2_\Lambda$-normalized field
\begin{equation}\label{eq:phi2}
\bs \varphi_1^{(2)}=\bs \phi_{1,v^{(2)}} - \bs \phi_{1,v^{(3)}}=
\frac{1}{\sqrt{1+A^2}}\begin{pmatrix}
A \begin{pmatrix}
-\cos \theta \sin\bra*{\sin \theta x_1-\cos \theta x_2}- \sin \bra*{x_2}\\
-\sin \theta \sin\bra*{\sin \theta x_1-\cos \theta x_2}
\end{pmatrix} \\
\cos\bra*{\sin \theta x_1-\cos \theta x_2}-\cos (x_2)
\end{pmatrix}.
\end{equation} 

\underline{Case 3}: Hexagonal lattice, $\abs{\tau}=1$, $\theta = \frac{\pi}{3}$

The wave vectors \eqref{eq:k_nu} corresponding to $k = (1,0),(0,1), (1,1)$ are
\[
v^{(4)} = \begin{pmatrix} \frac{\sqrt{3}}{2} \\ -\frac{1}{2} \end{pmatrix}, \quad
v^{(5)} = \begin{pmatrix}0 \\ 1 \end{pmatrix}, \quad
v^{(6)} = \begin{pmatrix}\frac{\sqrt{3}}{2} \\ \frac{1}{2}\end{pmatrix}.
\]
Therefore, taking into account \eqref{eq:phi_sym},
\[
\mathrm{ker}\,L_0=\mathrm{span } \set{\bs \phi_{i,v^{(j)}}, \; i=1,2, \, j=4,5,6}.
\]
Similarly, we have 
\[
\mathrm{Fix}_{X_0}(\Sigma_1)=\mathrm{span } \set{\bs \varphi_1^{(1)}} \quad \text{with} \quad \bs \varphi_1^{(1)}=\sqrt{2}\bs \phi_{1,v^{(4)}}
\]
and
\[
\mathrm{Fix}_{X_0}(\Sigma_2)=\mathrm{span } \set{\bs \varphi_1^{(2)}} \quad  
\text{with} \quad  \bs \varphi_1^{(2)}=\bs \phi_{1,v^{(4)}}-\bs \phi_{1,v^{(5)}}.
\]
The fixed subspace of $\Sigma_3$ (see \eqref{eq:Sigma3}) in the kernel $X_0$ is
\[\mathrm{Fix}_{X_0}(\Sigma_3)=\mathrm{span } \set{\bs \varphi_1^{(3)}}\]  
with the $L^2_\Lambda$-normalized field

\begin{equation}\label{eq:phi3}
\begin{aligned}
\bs \varphi_1^{(3)}&=\sqrt{\frac{2}{3}}(\bs \phi_{1,v^{(4)}}+\bs \phi_{1,v^{(5)}}+\bs \phi_{1,v^{(6)}})\\
&=\sqrt{\frac{2}{3(A^2+1)}}\begin{pmatrix}
A \begin{pmatrix}
-\frac{1}{2} \sin\bra*{\frac{\sqrt{3}}{2}x_1-\frac{1}{2}x_2}+\sin\bra*{x_2}+\frac{1}{2} \sin\bra*{\frac{\sqrt{3}}{2}x_1+\frac{1}{2}x_2}\\
-\frac{\sqrt{3}}{2} \sin\bra*{\frac{\sqrt{3}}{2}x_1-\frac{1}{2}x_2}-\frac{\sqrt{3}}{2} \sin\bra*{\frac{\sqrt{3}}{2}x_1+\frac{1}{2}x_2} 
\end{pmatrix}\\
\cos\bra*{\frac{\sqrt{3}}{2}x_1-\frac{1}{2}x_2}+\cos\bra*{x_2}+\cos\bra*{\frac{\sqrt{3}}{2}x_1+\frac{1}{2}x_2}
\end{pmatrix}.
\end{aligned} 
\end{equation}
\end{proof}

\begin{remark}
The first order bifurcation solution on non-equilateral lattices is an exact solution of the nonlinear equation at $\beta=0$, i.e.
\[
F(s\bs \phi_1, \lambda_0 -\alpha s^2)=0 \quad\text{for any }\quad s \in \R.
\]
\end{remark}

\begin{proof}[Proof of Theorem \ref{thm:main}]
Clearly $F(0,\lambda)=0$ for all $\lambda \in \R$. According to Proposition \ref{prop:fredholm}, the fixed subspace of the symmetry group $\Sigma=\Sigma_i$, $i=1,2,3$, is one-dimensional in all three cases 
\[
\mathrm{Fix}_{X_0}(\Sigma)=\mathrm{span } \set{\bs \varphi_1}, 
\]
where $\bs \varphi_1=\bs \varphi_1^{(i)}$ and
\[
D_\lambda D_{\m} F(0,\lambda_0) \scp{\bs \phi} = \bs \phi \notin \mathrm{ran}\, L_0 \quad\text{for any nonzero }\bs \phi \in \mathrm{Fix}_{X_0}(\Sigma).
\]
Invoking the equivariant branching lemma 
(see \cite{Chossat_Lauterbach, golubitsky2012singularities})
we conclude the existence of bifurcation solutions in the form of
\[
\lambda_s = \lambda_0+\varphi_\lambda(s), \quad \m_s =s\bs \varphi_1+\varphi_{\m}(s)
\] 
where $s \in (-\delta,\delta)$ for some $\delta>0$, $\varphi_\lambda:(-\delta,\delta) \to \R$ and 
$\varphi_{\m}:(-\delta,\delta) \to H_\Lambda^2$ are analytic in $s$ satisfying $\varphi_\lambda(0)=0$, $\varphi_{\m}(0)=0$ and $\scp{\varphi_{\m}(s),\bs\varphi_1}=0$.

Inserting the analytic expansion of bifurcation solutions into \eqref{eq:main} and \eqref{eq:averageEnergy}
we obtain \eqref{eq:BifurcationSolution}-\eqref{eq:energy}. Explicit calculations are carried out in Appendix \ref{sec:higherorder}.
\end{proof}

\paragraph{Topology of bifurcation solutions}
On hexagonal lattices, the bifurcation solution corresponding to the isotropy subgroup $\Sigma_3$ is nowhere vanishing and features in every primitive cell a skyrmion, i.e. a vortex-like structure with the magnetizations pointing upwards at the core and downwards at the perimeter, see Figure \ref{fig:lattice_top}(a).

On equilateral lattices, the horizontal component of $\bs \varphi_1^{(2)}$ has a finite number of isolated zeros in a primitive cell and forms a vortex or an antivortex around each zero. 
The antivortices are half-skyrmions (sometimes referred to as merons, see e.g. \cite{yu2018transformation}) and have magnetizations pointing upwards or downwards at the core;
while the center of vortices are singularity points ($\m=0$) due to the continuity and the $\Sigma_2$-invariance of bifurcation solutions, see Figure \ref{fig:lattice_top}(b).

\begin{figure}[hbtp]
 \centering
 \hfill
 \begin{subfigure}[b]{0.45\textwidth}
 \centering
 \includegraphics[width=0.8\linewidth]{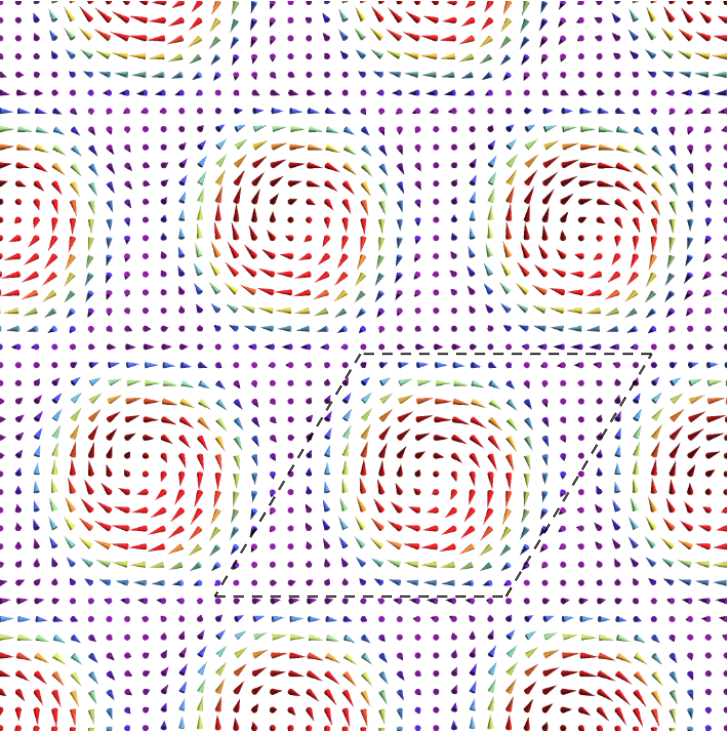}
 \caption{}
\end{subfigure}%
\hfill
\begin{subfigure}[b]{0.45\textwidth}
 \centering
 \includegraphics[width=0.8\linewidth]{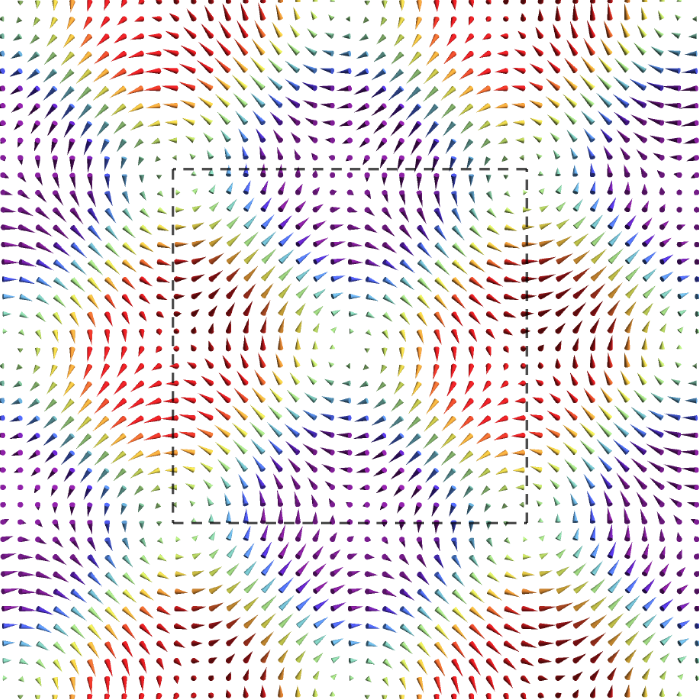}
 \caption{}
\end{subfigure}
\caption{Top view of (a) $\bs \varphi_1^{(3)}$ on hexagonal lattice and (b) $\bs \varphi_1^{(2)}$ on square lattice. The cones indicates the direction and length of the in-plane magnetizations; the color indicates the out-of-plane magnetizations (red: upwards, green: in-plane, blue: downwards). The dashed lines in (a) and (b) define a primitive cell of the hexagonal and square lattice, respectively.}\label{fig:lattice_top}
\end{figure}

\section{Linear stability of the bifurcation solutions}\label{sec:stability}

In this section we discuss the stability of bifurcation solutions following perturbation methods as e.g. in \cite{kielhofer2011bifurcation}.
Suppose $(\m_s, \lambda_s)$ is a bifurcation solution as in Theorem \ref{thm:main} with 
\[
\lambda_0= - 1 - \frac{\beta}{2}+ \sqrt{4\kappa^2+\frac{\beta^2}{4}}.
\]

We focus on the larger root as positivity turns out to be necessary for the stability of bifurcation solutions.
The linearization of \eqref{eq:main} at $(\m_s, \lambda_s)$ 
\[L_s=D_{\m} F(\m_s,\lambda_s):H^2_{\Lambda} \subset L^2_{\Lambda} \to L^2_{\Lambda}
\] 
is given by 
\begin{align*}
L_s \bs \phi &= -\lap \bs \phi + 2\kappa \grad \times \bs \phi +\lambda_s \bs \phi+\alpha\bra*{\abs{\m_s}^2\bs \phi +2 (\m_s\cdot \bs \phi)\m_s} + \beta \phi_3 \ein_3.
\end{align*}
We first investigate the spectrum of $L_0$. 
\begin{lemma}\label{la:L0} 
The spectrum of $L_0$ is discrete and non-negative if \eqref{eq:lambda0+} and \eqref{eq:stability_kappa} hold. Otherwise $L_0$ has negative eigenvalues. 
\end{lemma}
\begin{proof}
The equation $(L_0-\mu)\phi=0$ admits constant nonzero solutions at $\mu=\lambda_0$. So $L_0$ has negative eigenvalues with constant eigenfunctions
precisely if $\lambda_0 < 0$. 
As in Lemma \ref{la:lambda}, the eigenvalues of $L_0$ with non-constant eigenfunctions at $\lambda_0$ are
\begin{align*}
\mu_{0,\omega \pm} = \abs{\omega}^2 - 1 + \sqrt{4\kappa^2+\frac{\beta^2}{4}} \pm \sqrt{4\kappa^2\abs{\omega}^2+\frac{\beta^2}{4}}, \quad \omega \in \Lambda^* \setminus \set{0}.
\end{align*}
We have $\abs{\omega} \ge 1$ for all $\omega \in \Lambda^* \setminus \set{0}$, which, together with \eqref{eq:stability_kappa}, ensures that  $\mu_{0,\omega+} \ge 0$ and
\[
\mu_{0,\omega-} 
=\bra*{\abs{\omega}^2 - 1}\bra*{1-\frac{4\kappa^2}{\sqrt{4\kappa^2+\frac{\beta^2}{4}}+\sqrt{4\kappa^2\abs{\omega}^2+\frac{\beta^2}{4}}}}\ge 0
\]
for all $\omega \in \Lambda^* \setminus \set{0}$ and any $\beta \ge 0$.
\end{proof}

From now on we focus on the case $\lambda_0 >0$. 
It follows from Lemma \ref{la:L0} and the standard perturbation theory of eigenvalue \cite{kato} that the spectrum of $L_s$ consists of eigenvalues of the same multiplicities in an neighbourhood of the eigenvalues of $L_0$. Thus the stability of $(\m_s,\lambda_s)$ depends on the perturbation of the critical eigenvalue 0. It follows from Proposition \ref{prop:fredholm} that there exist the following topological decompositions
\[
L^2_\Lambda = \mathrm{ker}\, L_0 \oplus \mathrm{ran}\, L_0
\quad \text{and} \quad 
H^2_\Lambda = \mathrm{ker}\, L_0 \oplus X_1, 
\]
where $X_1= \{ \bs \phi \in H^2_{\Lambda}: \bs \phi \perp \ker L_0 \text{ in } L^2\}$.

We first consider the perturbation of the zero eigenvalue corresponding to the eigenvector $\bs \varphi_1$ spanning $\mathrm{Fix}_{X_0}(\Sigma)$.
\begin{lemma}\label{la:Ls} 
After reducing $\delta>0$ if necessary, there exists a smooth map
\[
(\mu,\bs \psi): (-\delta,\delta) \to \R \times X_1
\quad \text{with} \quad (\mu(0),\bs\psi(0))=(0,0)
\]
such that
\begin{equation}\label{eq:perturbedEigenvalue1}
L_s(\bs \varphi_1 + \bs \psi(s)) = \mu(s) (\bs \varphi_1 + \bs \psi(s)).
\end{equation}
Furthermore, there exists a $C>0$ so that $\mu(s)=C\alpha s^2$ for $s \in (-\delta,\delta)$.
\end{lemma}
\begin{proof}
We introduce the smooth operator
\[
G: \bra*{\mathrm{Fix}_{X_0}(\Sigma) \oplus X_1} \times \R \times X_1 \times \R \to \mathrm{Fix}_{X_0}(\Sigma) \oplus \mathrm{ran}\, L_0
\]
given by
\[
G(\m_s, \lambda_s, \bs \psi, \mu)= L_s(\bs \varphi_1 + \bs\psi) - \mu(\bs \varphi_1 + \bs\psi).
\]
Since $L_s\bs \phi \in \mathrm{Fix}_{X_0}(\Sigma) \oplus \mathrm{ran}\, L_0$ for $\bs \phi \in \mathrm{Fix}_{X_0}(\Sigma)\oplus X_1$, this operator is well-defined.
As $G(0,\lambda_0,0,0)= 0$ and the differential
\[
\partial_{(\bs \psi, \mu)} G(0,\lambda_0,0,0): X_1 \times \R \to \mathrm{Fix}_{X_0}(\Sigma) \oplus \mathrm{ran}\, L_0
\]
given by
\[
\partial_{(\bs \psi, \mu)}G(0,\lambda_0,0,0) 
= \begin{pmatrix}L_0 \\ -\bs \varphi_1\end{pmatrix}
\]
is invertible, the implicit function theorem provides a smooth map
\[
(\mu, \bs\psi): (-\delta,\delta) \to \R \times X_1 \quad \text{with} \quad (\mu(0), \bs\psi(0))=(0,0)
\]
so that
 \[
G(\m_s,\lambda_s,\bs \psi(s), \mu(s))=0 \quad \text{for} \quad s \in (-\delta,\delta).
\]
It remains to examine the properties of $\mu(s)$.
Differentiating \eqref{eq:perturbedEigenvalue1} with respect to $s$ in $s=0$ yields
\begin{equation}\label{eq:mu1}
\begin{aligned}
\dot{\mu}(0)\bs \varphi_1=&D_{\m}^2 F(\m_0,\lambda_0) \scp{\bs \varphi_1, \bs \varphi_1} + 
D_\lambda D_{\m} F(\m_0,\lambda_0)\scp{\bs \varphi_1}\dot{\lambda}_s |_{s=0}\\
& \quad +
D_{\m} F(\m_0,\lambda_0)\scp{\dot{\bs \psi}(0)}
= L_0\dot{\bs \psi}(0)
\end{aligned}
\end{equation}
since $\dot{\lambda}_s = 2s \nu_2 + O(s^2)$ and 
\begin{align*}
D_{\m}^2 F(\m_0,\lambda_0) \scp{\bs \varphi, \bs \psi} =& 2\alpha \bra*{(\bs \varphi \cdot \bs \psi) \m_0+(\m_0\cdot \bs \varphi)\bs \psi+(\m_0\cdot \bs \psi)\bs \varphi}
=0
\end{align*}
with $\m_0=0$. Testing \eqref{eq:mu1} with $\bs \varphi_1$ yields 
\begin{align*}
\dot{\mu}(0) = \frac{\scp{L_0\dot{\bs \psi}(0), \bs \varphi_1}}{\scp{\bs \varphi_1,\bs \varphi_1}} =0.
\end{align*} 
Calculating the second derivative of $\mu(s)$ at $s=0$, we obtain
\begin{align*}
\ddot{\mu}(0)\bs \varphi_1 &= D_{\bs m}^3 F(\m_0,\lambda_0) \scp{\bs \varphi_1, \bs \varphi_1, \bs \varphi_1}+D_{\m} F(\m_0,\lambda_0)\scp{\ddot{\bs \psi}(0)} + D_\lambda D_{\m} F(\m_0,\lambda_0)\scp{\bs \varphi_1}\ddot{\lambda}(0).
\end{align*}
Taking into account
\begin{align*}
& D_\lambda D_{\m} F(\m_0,\lambda_0)\scp{\bs \varphi} = \bs \varphi, \quad
\ddot{\lambda}(0)=2\nu_2 \\ 
& D_{\m}^3 F(\m_0,\lambda_0)\scp{\bs \varphi, \bs \phi, \bs \psi}= 2\alpha\bra*{(\bs \varphi \cdot \bs \phi) \bs \psi + (\bs \varphi \cdot \bs \psi) \bs \phi + (\bs \phi \cdot \bs \psi) \bs \varphi }
\end{align*}
and using Theorem \ref{thm:main} yields
\begin{align*}
\ddot{\mu}(0) &= \frac{6\alpha\scp{\abs{\bs \varphi_1}^2\bs \varphi_1,\bs \varphi_1}+2\nu_2\scp{\bs \varphi_1,\bs\varphi_1}}{\scp{\bs \varphi_1,\bs \varphi_1}}= 4 \alpha \scp{\abs{\bs \varphi_1}^4}.
\end{align*}
Hence $\mu(s) = C\alpha s^2$ with a positive constant $C$. Provided $\alpha>0$, we have $\mu(s)>0$ for nonzero $s \in (-\delta,\delta)$.
\end{proof}
%
The helical solution on non-equilateral lattices is stable.
\begin{proposition}\label{prop:Stability_nonequi}
For $\beta \ge 0$ and every nonzero $s \in (-\delta,\delta)$, the bifurcation solution on non-equilateral lattices is linearly stable in the sense that
\[ L_s \ge 0 \quad\text{with}\quad \ker L_s = \mathrm{span}\set{\partial_2 \m_s}.\]
\end{proposition}
\begin{proof}
Recall that on non-equilateral lattices
\[
\mathrm{ker}\,L_0=\mathrm{span } \set{\bs \phi_{1,v^{(1)}}} \oplus \mathrm{span } \set{ \bs \phi_{2,v^{(1)}}},
\] 
where both $\bs \varphi_1^{(1)} = \bs \phi_{1,v^{(1)}}$ and $\bs \phi_{2,v^{(1)}}$ depend only on the spatial variable $x_2$ and 
$\del_2 \bs \varphi_1^{(1)} = \bs \phi_{2,v^{(1)}}$ By translational invariance $\partial_2 \m_s=s\partial_2 \bs \varphi_1 + O(s^3)$ is, for small $s$, a non-trivial element of $\ker L_s$, and the claim follows with
Lemma \ref{la:Ls}.
\end{proof}

The quadratic vortex-antivortex lattice is stable under large enough anisotropy.
\begin{proposition}\label{la:Stability_equi}
For every nonzero $s \in (-\delta,\delta)$, the vortex-antivortex bifurcation solution on square lattices is linearly stable in the sense that
\[ L_s \ge 0 \quad\text{with}\quad \ker L_s = \mathrm{span}\set{\partial_1 \m_s, \partial_2 \m_s}.\]
 provided $\alpha>0$ and
\[
\beta>\frac{4}{\sqrt{3}}\kappa \approx 2.3 \kappa.
\]
For $\beta<\frac{4}{\sqrt{3}}\kappa$ there exists $\bs\phi \in L^2_\Lambda$ such that $\scp{L_s\bs \phi, \bs \phi} < 0$. 
\end{proposition}
\begin{proof}
Recall that on square lattices 
\begin{align*}
\ker \,L_0 &=\mathrm{span } \set{\bs \phi_{1,v^{(2)}}, \bs \phi_{1,v^{(3)}}, \bs \phi_{2,v^{(2)}}, \bs \phi_{2,v^{(3)}}}\\
&=
\mathrm{span } \set{\bs \phi_{1,v^{(2)}} -\bs \phi_{1,v^{(3)}}} \oplus \mathrm{span } \set{\bs \phi_{1,v^{(2)}} +\bs \phi_{1,v^{(3)}}} \oplus \mathrm{span } \set{\bs \phi _{2,v^{(2)}},\bs \phi _{2,v^{(3)}}}.
\end{align*}
Note that $\bs\varphi_1^{(2)} =\bs \phi_{1,v^{(2)}} -\bs \phi_{1,v^{(3)}}$ is $\Sigma_2$-invariant, while $\tilde{\bs\varphi_1} =\bs \phi_{1,v^{(2)}} +\bs \phi_{1,v^{(3)}}$ is invariant under another symmetry group $\tilde{\Sigma}= \set{R_k, \, k=0,1,2,3}$ where the associated $SO(3)$ elements are
\[
\bs{R}_k = \begin{pmatrix}
\cos \frac{k\pi}{4} & - \sin \frac{k\pi}{4} & 0 \\
\sin \frac{k\pi}{4} & \cos \frac{k\pi}{4} & 0 \\
0 & 0 & 1
\end{pmatrix}.
\]
Repeating the argument from Lemma \ref{la:Ls}, there exists a smooth map
\[
(\tilde\mu,\tilde{\bs \psi}): (-\delta,\delta) \to \R \times X_1
\quad \text{with} \quad (\tilde\mu(0),\tilde{\bs\psi}(0))=(0,0)
\]
such that
\[
L_s(\tilde{\bs \varphi}_1 + \tilde{\bs \psi}(s)) = \tilde\mu(s) (\tilde{\bs \varphi}_1 + \tilde{\bs \psi}(s)).
\]
Moreover, we obtain $\tilde\mu(s)=\tilde C \alpha s^2$ with
\[
\tilde C=4\scp{(\tilde{\bs \varphi_1} \cdot \bs \varphi_1^{(2)})^2}+2\scp{\abs{\tilde{\bs \varphi_1}}^2 \abs{ \bs \varphi_1^{(2)}}^2}+2\frac{\nu_2}{\alpha} \scp{\abs{\tilde{\bs \varphi_1}}^2}
= \frac{A^2-3}{A^2+1}
\]
where $A$ is the amplitude given in \eqref{eq:linearSol} depending on $\kappa$ and $\beta$. $\tilde C$ is positive provided $A^2>3$, which is fulfilled if $\beta>\frac{4}{\sqrt{3}}\kappa$.

Finally $\partial_i \m_s=s\partial_i \bs \varphi_1 + O(s^3)$, $i=1,2$, are linearly independent for small $s$ due to the linear independence of $\partial_1 \bs \varphi_1^{(2)}=\bs \phi _{2,v^{(2)}}$ and $\partial_2 \bs \varphi_1^{(2)}=-\bs \phi _{2,v^{(3)}}$, and annihilate $L_s$ by translational invariance of $F$. We conclude that $L_s \ge 0$ and 
\[\ker L_s = \mathrm{span}\set{\partial_1 \m_s, \partial_2 \m_s}.\]

If $\beta<\frac{4}{\sqrt{3}}\kappa$, then $\scp{L_s \bs \phi, \bs \phi} < 0$ for any $\bs\phi \in \mathrm{Fix}_{X_0}(\tilde\Sigma)$.
\end{proof}

The same argument proves that the helical bifurcation solution on square lattice is stable if $\beta<\frac{4}{\sqrt{3}}\kappa$ and unstable for $\beta > \frac{4}{\sqrt{3}}\kappa$, see \cite{Li2020}.

The hexagonal skyrmion lattice is unstable independently of any additional easy-plane anisotropy: For example, for any nonzero $s \in (-\delta,\delta)$
\[
\scp{L_s \bs \phi, \bs \phi}=-\frac{\alpha(2A^2+3)}{3(A^2+1)}s^2 + O(s^4)<0
\]
for $\bs \phi= \bs \phi_{1,v^{(4)}}-\bs \phi_{1,v^{(5)}}$. Similarly it can be shown that the vortex-antivortex bifurcation solution on hexagonal lattice is unstable and the helical bifurcation solution is linearly stable under any easy-plane anisotropy, for details see \cite{Li2020}.
\section{Numerical simulations}\label{sec:num}

\subsection{Numerical scheme}

To examine critical points we consider the $L^2$-gradient flow equation for the energy functional $E_\Lambda$ 
\begin{equation}\label{eq:gradientflow}
\partial_t \m + \mathrm{grad}_{L^2} E_\Lambda(\m)=0.
\end{equation}
Decomposing the energy gradient into a linear (second-order elliptic) 
\[
\mathcal{L} \m := -\lap \m +2\kappa \grad \times \m + \lambda \m + \beta m_3 \ein_3
\]
and a nonlinear part 
\[
\mathcal{N}(\m) := -\alpha \abs{\m}^2 \m,
\]
the gradient flow equation \eqref{eq:gradientflow} reads
\begin{equation}\label{eq:gradientFlow}
\partial_t \m + \mathcal{L} \m = \mathcal{N}(\m).
\end{equation}
We aim to find equilibria of the energy functional $E_\Lambda$ by solving \eqref{eq:gradientFlow} numerically on a primitive cell $\Omega_{\Lambda}$ induced by lattice spanned by $\set{2\pi\ein_1 , 2\pi \tau }$. Eq. \eqref{eq:gradientFlow} is discretized by a modified Crank-Nicolson approximation for the time variable and a Fourier collocation method for the space variable. We denote $\m_N$ the trigonometric interpolation function of $\m$ on the discretized grid by $N^2$ collocation points
\[
x_{ij}:=2\pi \bra*{\frac{i}{N}, \frac{-\tau \cos \theta i + j}{N \tau \sin \theta}}, \quad (i,j) \in \N_N^2,
\]
for $\N_N=\set{0,\ldots,N-1}$, $N \in \N$ and odd. For continuous fields $\bs u,\bs v$ on $\Omega_{\Lambda}$, we define the discrete $L^2$ scalar product
\[
\scp{\bs u,\bs w}_N:=\bra*{\frac{2\pi}{N}}^2\frac{1}{\abs{\Omega_{\Lambda}}}\sum_{i=0}^{N-1} \sum_{j=0}^{N-1} \bs u(x_{ij})\cdot \bs w(x_{ij}),
\]
the associated norm $\norm{\cdot}_N$, defined by $\norm{\bs u}^2_N := \scp{\bs u,\bs u}_N$ and the discrete energy
\[
E_N(\m):=\frac{1}{2}\scp{\m, \mathcal{L}\m}_N + \scp{1, W(\m)}_N \quad \text{with} \quad W(\m)= \frac{\alpha}{4} |\m|^4.
\]
Our numerical scheme at time iteration $n+1$ reads: find $\m^{n+1}$ such that
\[
\frac{\m^{n+1}_N-\m^n_N}{\Delta t} + \mathcal{L} \bra*{\frac{\m^{n+1}_N+\m^n_N}{2}} = I_N\mathcal{N}(\m^n_N,\m^{n+1}_N),
\]
where $I_N$ denotes the trigonometric interpolation operator and 
\[
\mathcal{N}(\bs u,\bs w)=-\alpha\frac{\bs u+\bs w}{4}\bra*{\abs{\bs u}^2 + \abs{\bs w}^2}.
\]
At each time iteration, in order to find the solution $\m^{n+1}_N$, we use a fixed point iteration
\[
\m^{n+1}_{N,k+1} = \bra*{\mathrm{Id}+\frac{\Delta t}{2}\mathcal{L}}^{-1}\pra*{\bra*{\mathrm{Id}-\frac{\Delta t}{2}\mathcal{L}}\m^n_N + \Delta t \mathcal{N}(\m_{N,k}^{n+1},\m^n_N)}
\]
for some time step $\Delta t$. Well-posedness and a-priori error bounds of this numerical scheme follow analogously as in \cite{condette2011spectral}. Given
initial data $\m_N^0$, the corresponding sequence $(\m_N^n)_{n \in \N_0}$ satisfies the following energy law
\[
\frac{1}{\Delta t} \norm{\m_N^{n+1}-\m_N^n}_N^2 + E_N(\m_N^{n+1}) = E_N(\m_N^n), \quad n \in \N_0.
\]
We have implemented this numerical scheme in MATLAB. At each time-step the iteration process stops if a certain norm of the difference of two successive iterations becomes smaller than a chosen stopping tolerance. In our case we choose the $L^\infty$-norm and set the stopping tolerance to $10^{-8}$. The discrete energy is evaluated in each time step, and the terminal time is controlled through a smallness condition for the discrete energy gradient, i.e.
\[
\frac{E(\m_N^n)-E(\m_N^{n+1})}{\Delta t}< 10^{-7}.
\]
After the termination of the scheme, an equilibrium configuration is reached approximately.

\subsection{Numerical experiments}\label{sec:numexp}

We have implemented the method on a lattice of $275 \times 275$ grid points and with a time increment $\Delta t=0.1$ for different parameters and a randomly distributed 
initial field with modulus between 0 and 0.1 as initial condition. 

\paragraph{Parameter study and assessment of the stability condition \eqref{eq:lambda0+}.} First we implemented the simulations on a square lattice for different parameters $\kappa$ and $\beta$ near the bifurcation point by setting $\lambda = \lambda_0 + \delta\nu_2$, where $\delta=0.01$, $\lambda_0$ and $\nu_2$ are calculated according to \eqref{eq:lambda0+} and \eqref{eq:lambda2}, respectively. 

For $\kappa \in (0,0.5)$, the bifurcation point $\lambda_0$ is negative for any $\beta \ge 0$. In this case, the stability condition \eqref{eq:lambda0+} is not fulfilled we obtained an almost homogeneous field.

When the value of $\kappa$ was increased over $0.5$, vortex-antivortex lattice configurations were observed for $\beta \ge 0$ small enough so that $\lambda_0>0$. For $\beta$ large enough so that $\lambda_0<0$, the vortex-antivortex lattice configuration decayed to an almost homogenous field, as shown in Figure \ref{fig:kappa0608}.

\begin{figure}[htbp]
 \centering
 \hfill
 \begin{subfigure}[b]{0.22\textwidth}
 \centering
 \includegraphics[width=\textwidth]{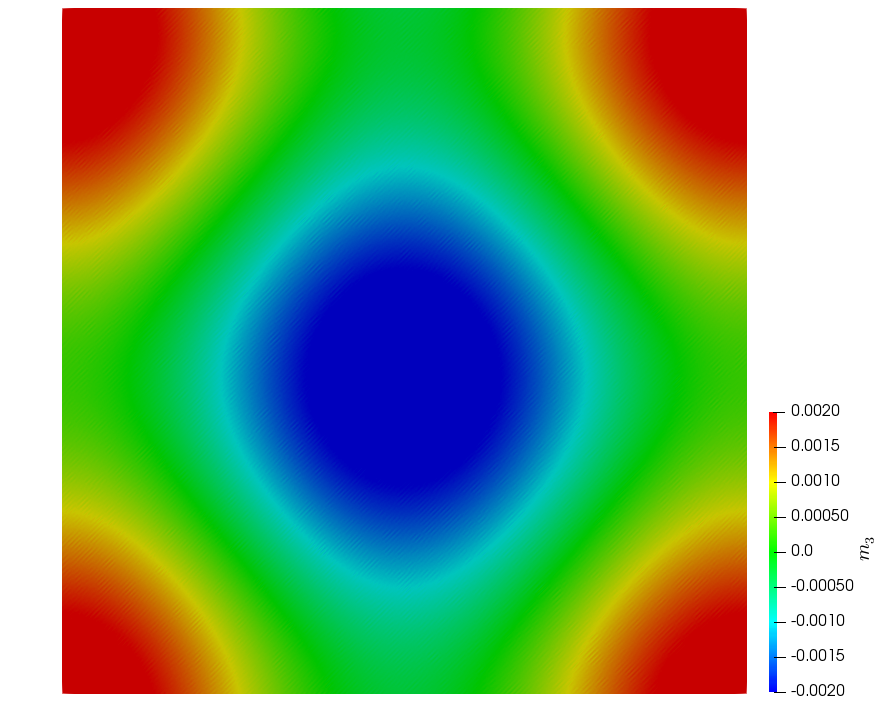}
 \caption{}
 \end{subfigure}
 \hfill
 \begin{subfigure}[b]{0.22\textwidth}
 \centering
 \includegraphics[width=\textwidth]{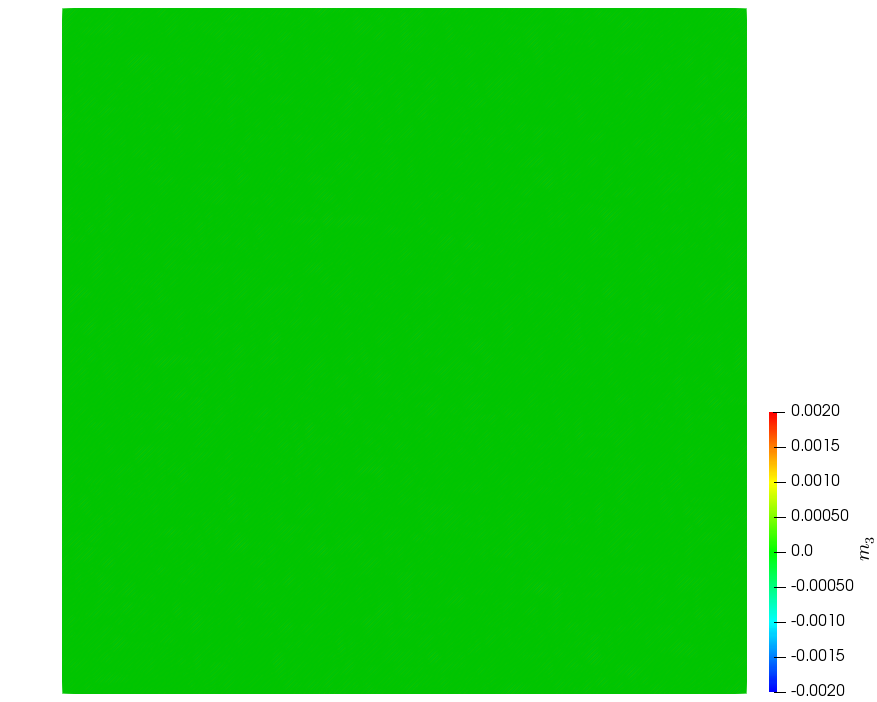}
 \caption{}
 \end{subfigure}
 \hfill
 \begin{subfigure}[b]{0.22\textwidth}
 \centering
 \includegraphics[width=\textwidth]{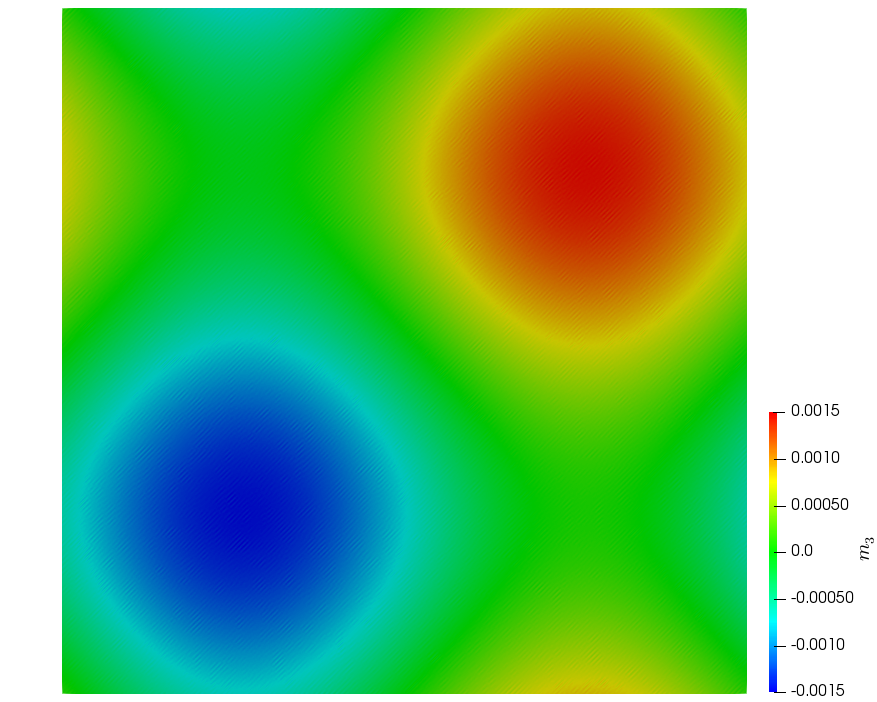}
 \caption{}
 \end{subfigure}
 \hfill
 \begin{subfigure}[b]{0.22\textwidth}
 \centering
 \includegraphics[width=\textwidth]{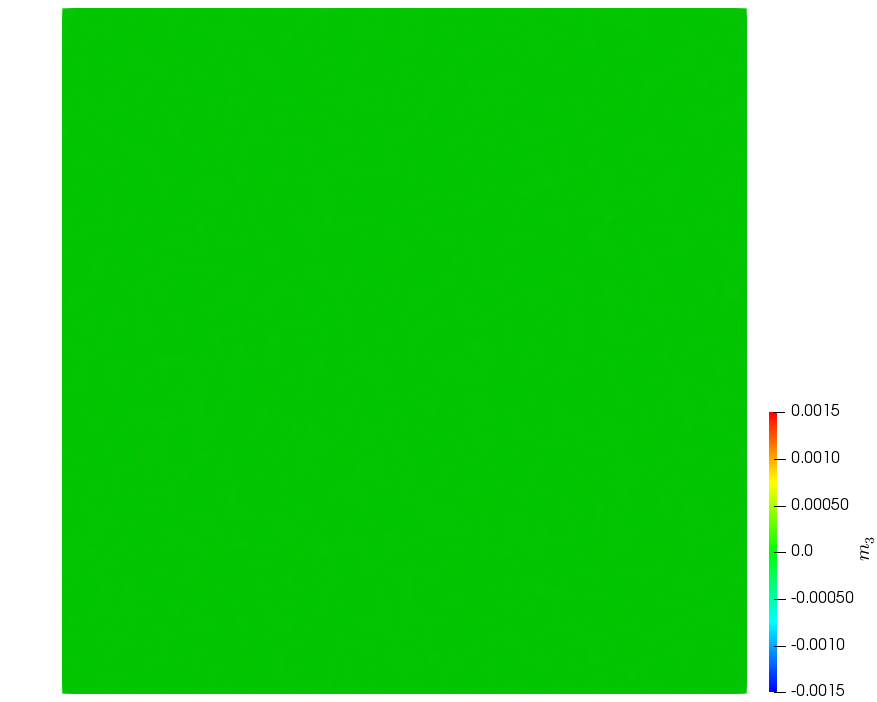}
 \caption{}
 \end{subfigure}
 \caption{Simulation results on square lattice for $\alpha=1$ and different values of $\kappa$ and $\beta$, respectively (a) $\kappa=-0.6$, $\beta=0$, $\lambda=0.1875$, (b) $\kappa=-0.6$, $\beta=1$, $\lambda=-0.2125$, (c) $\kappa=0.8$, $\beta=1$, $\lambda=0.1638$, (d) $\kappa=0.8$, $\beta=2$, $\lambda=-0.1257$.}
 \label{fig:kappa0608}
\end{figure}

Other patterns emerged for $\kappa> 1.2$ and small $\beta \ge 0$. For example, at $\kappa=1.4$ and small $\beta \ge 0$ the solution converged to a stripe pattern, i.e. helices with a pitch smaller than $2\pi$. 
Vortex-antivortex lattice configurations were observed for $\beta$ in the admissible region, see Figure \ref{fig:para_range}, i.e. $\beta$ larger than the stability threshold (about $2.3\kappa$) and $\lambda_0>0$. Increasing $\beta$ further so that $\lambda_0<0$, we obtained the almost homogeneous field again, as shown in Figure \ref{fig:kappa14}. 

\begin{figure}[htbp]
 \centering
 \hfill
 \begin{subfigure}[b]{0.22\textwidth}
 \centering
 \includegraphics[width=\textwidth]{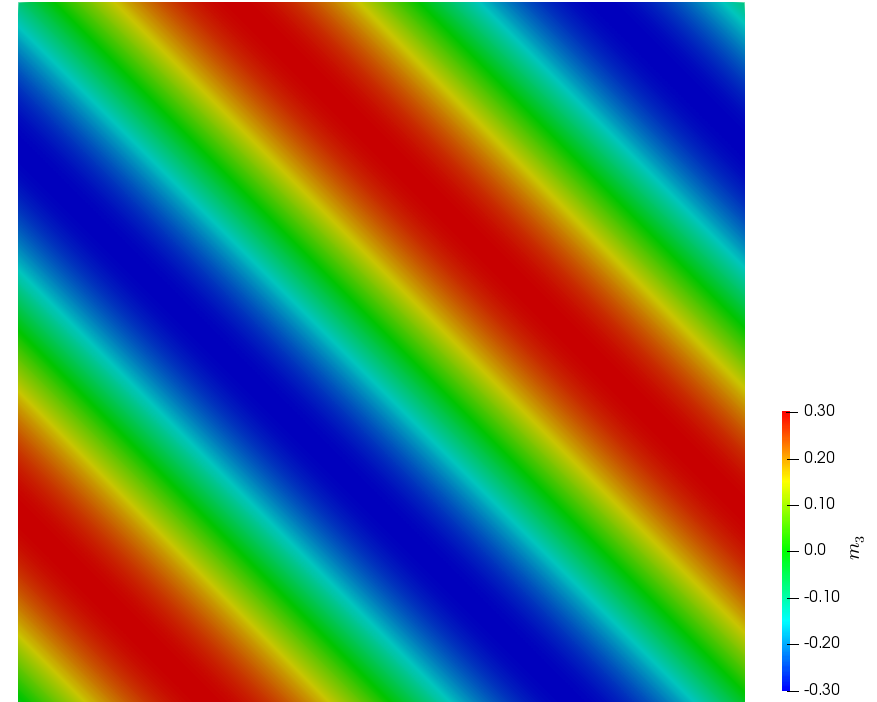}
 \caption{$\beta=3$}
 \end{subfigure}
 \hfill
 \begin{subfigure}[b]{0.22\textwidth}
 \centering
 \includegraphics[width=\textwidth]{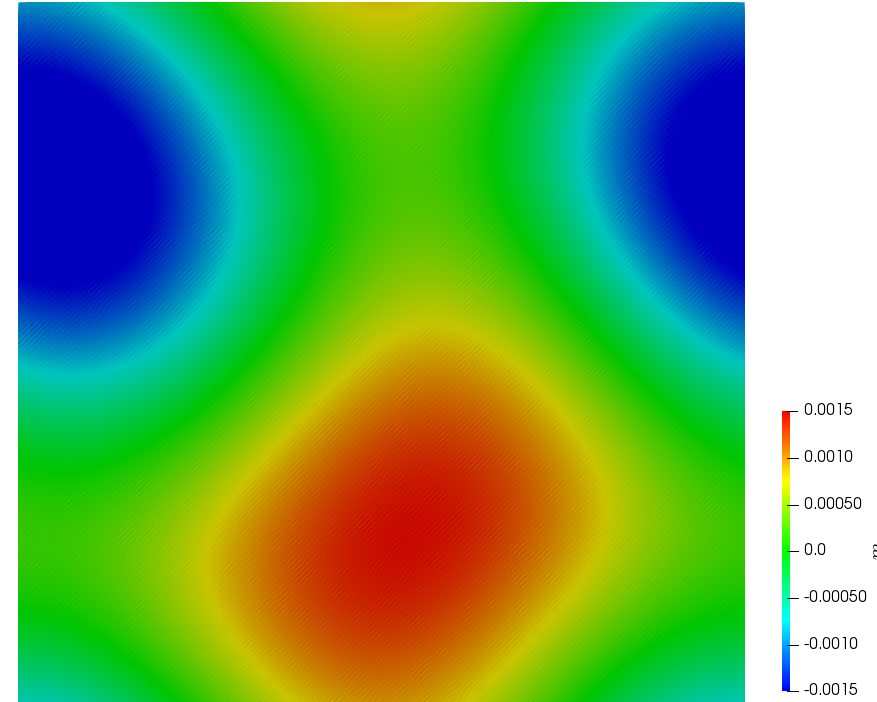}
 \caption{$\beta=4$}
 \end{subfigure}
 \hfill
 \begin{subfigure}[b]{0.22\textwidth}
 \centering
 \includegraphics[width=\textwidth]{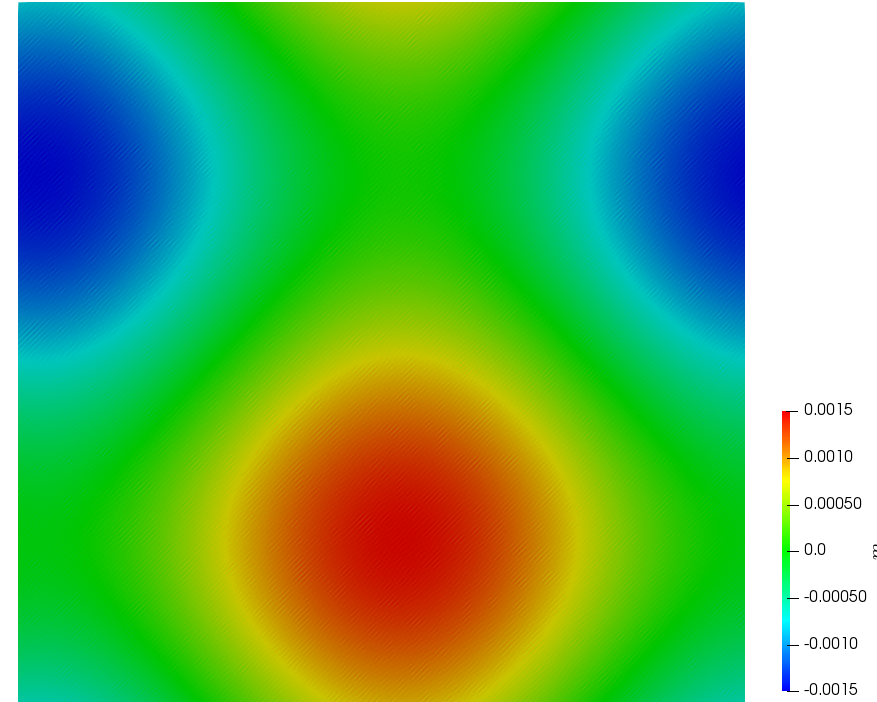}
 \caption{$\beta=5$}
 \end{subfigure}
 \hfill
 \begin{subfigure}[b]{0.22\textwidth}
 \centering
 \includegraphics[width=\textwidth]{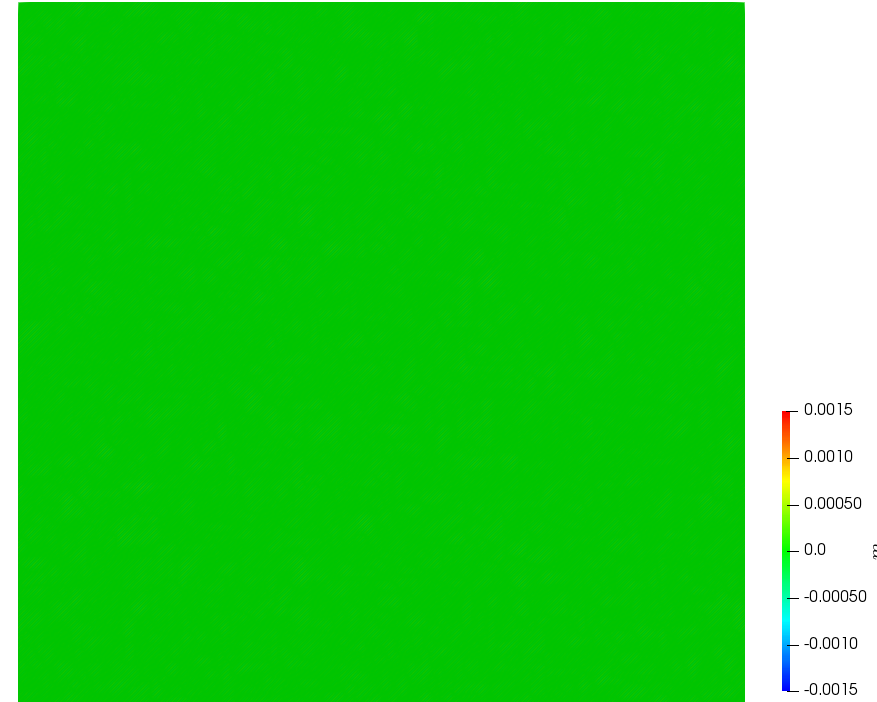}
 \caption{$\beta=7$}
 \end{subfigure}
 \caption{Simulation results on square lattice for $\kappa=1.4$, $\alpha=1$ and different values of $\beta$, respectively (a) $\beta=3$, $\lambda=0.6640$, (b) $\beta=4$, $\lambda=0.4284$, (c) $\beta=5$, $\lambda=0.2412$, (d) $\beta=7$, $\lambda=-0.0303$.}
 \label{fig:kappa14}
\end{figure}

\paragraph{Stability of vortex-antivortex solution under the perturbation of lattices} 
We proved the stability of quadratic vortex-antivortex solutions in certain parameter region under $\Lambda$-periodic perturbations (see Section \ref{sec:stability}). 
Complementarily we examine the persistence of quadratic vortex-antivortex solutions under the perturbations of lattice shape by implementing the simulations on different lattices with the same parameters. For small perturbations of $\tau=e^{i\pi/2}$ we obtained the vortex-antivortex lattice configuration, while for $\abs{\tau}$ large enough only the helix state was observable, as shown in Figure \ref{fig:tau}.

\begin{figure}[htbp]
 \centering
 \hfill
 \begin{subfigure}[b]{0.3\textwidth}
 \centering
 \includegraphics[width=\textwidth]{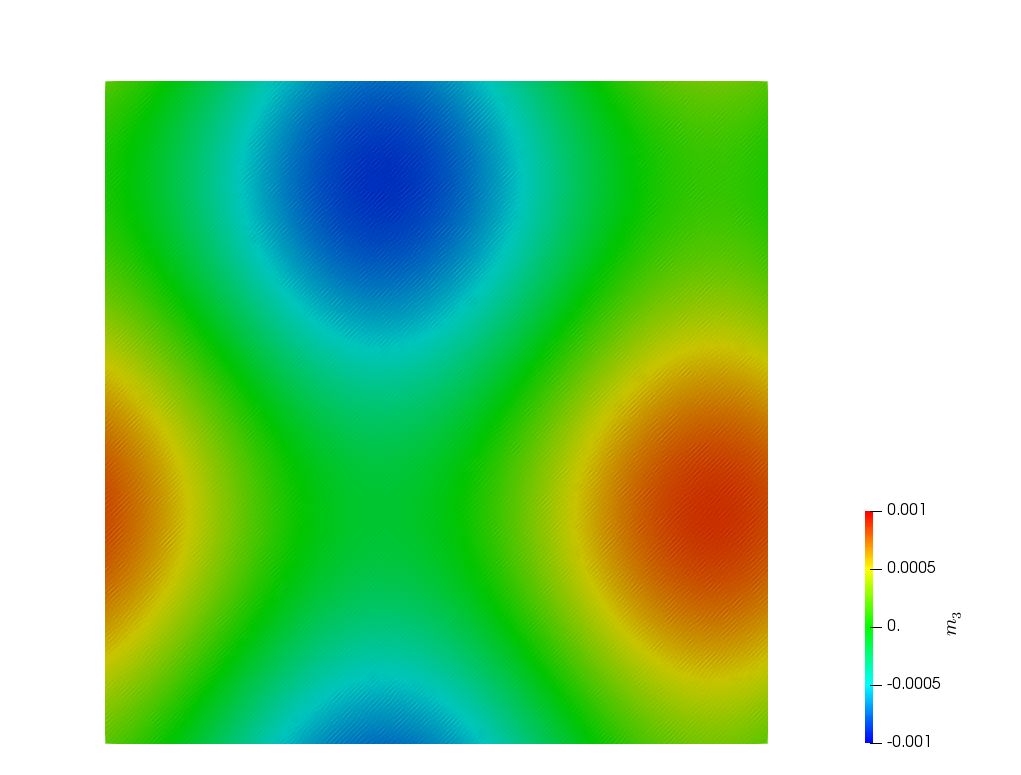}
 \caption{ }
 \end{subfigure}
 \hfill
 \begin{subfigure}[b]{0.3\textwidth}
 \centering
 \includegraphics[width=\textwidth]{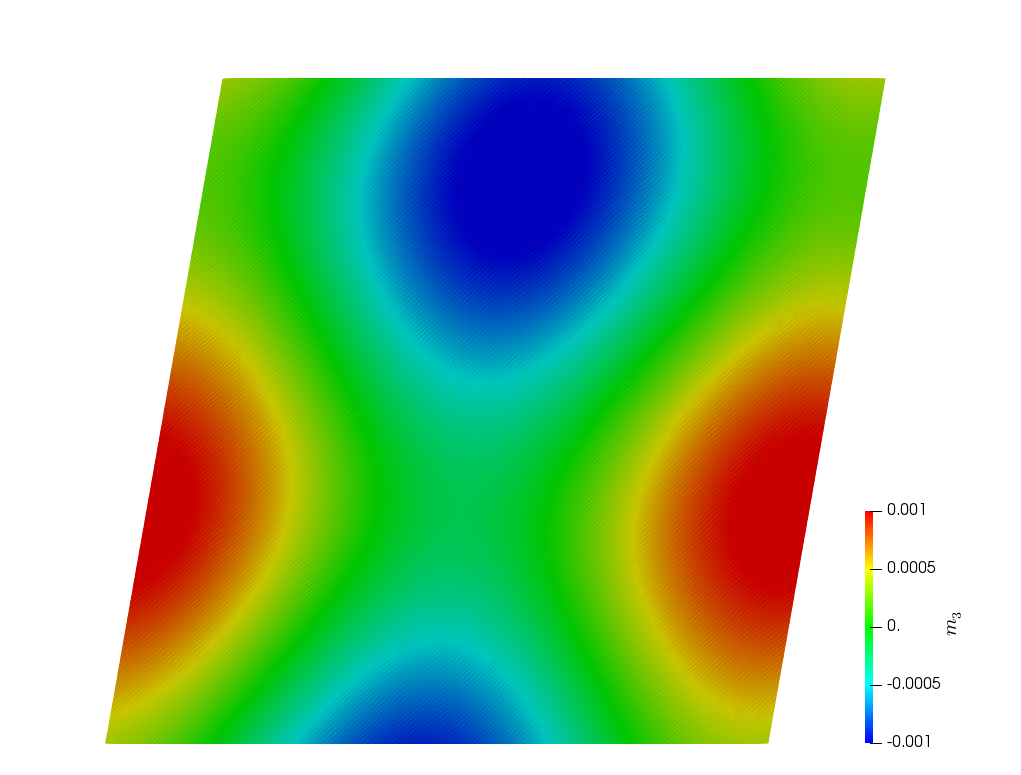}
 \caption{ }
 \end{subfigure}
 \hfill
 \begin{subfigure}[b]{0.3\textwidth}
 \centering
 \includegraphics[width=\textwidth]{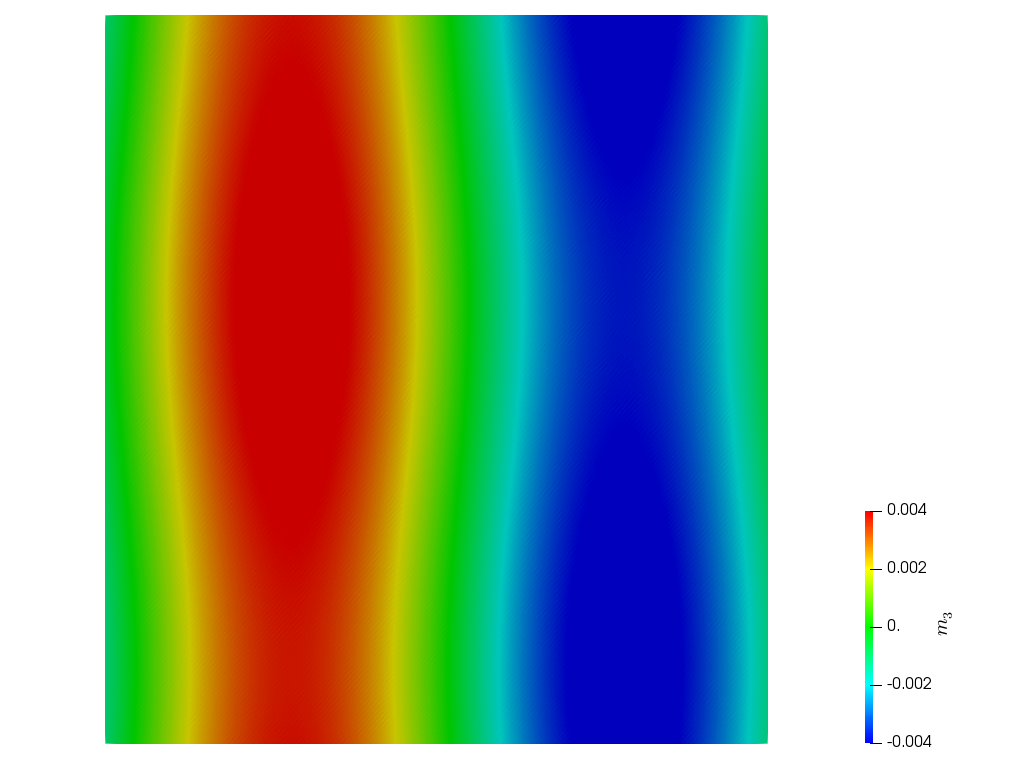}
 \caption{ }
 \end{subfigure}
 \caption{Simulation results on lattices $\abs{\tau} e^{i \theta}$ for $\kappa=1.4$, $\alpha=1$, $\beta=5$, $\lambda=0.2412$ and different values of $\tau$, respectively (a) $\abs{\tau}=1.0$, $\theta=90^\circ$, (b) $\abs{\tau}=1.02$, $\theta=80^\circ$, (c) $\abs{\tau}=1.10$, $\theta=90^\circ$.}
 \label{fig:tau}
\end{figure}


\subsection*{Acknowledgements} 
This work is support by Deutsche Forschungsgemeinschaft (DFG grant no. ME 2273/3-1). 
\appendix

\section{Helical state}\label{sec:helix}

In this section we consider vector fields $\m$ defined on a domain in $\R^3$. Given $\bs{R} \in O(3)$ we write
\[
\m_{\bs{R}}(\x) := \bs{R}\m(\bs{R}^{T}\x) \quad \text{for} \quad \x \in \R^3.
\]

\begin{lemma} \label{lemma:curl}
For smooth $\m$ it holds that
\[
\nabla \times \m_{\bs{R}} = \det \bs{R} \; (\nabla \times \m)_{\bs{R}}
\]
and hence
\[
\m_{\bs{R}} \cdot (\nabla \times \m_{\bs{R}}) = \det \bs{R}  \; \left(  \m \cdot (\nabla \times \m) \right)  \circ \bs{R}^{T}. 
\]
\end{lemma}

\begin{proof} Summing over repeated indices we write the curl as
\[
\nabla \times \m_{\bs{R}} = \sum_{j=1}^3 \ein_j \times \del_j \m_{\bs{R}} 
\]
By the chain rule $\del_j \m_{\bs{R}}= R_{jk} (\del_k \m)_{\bs{R}}$ and with $\bs{R} \ein_k = \ein_j R_{jk}$ it follows
form the $O(3)$ skew-symmetry of the vector product that 
\[
\nabla \times \m_{\bs{R}}= \det \bs{R} \; (\nabla \times \m)_{\bs{R}}.
\]
The second claim follows from the isometry property of $\bs{R}$.
\end{proof}

\begin{lemma} \label{lemma:Beltrami}
Suppose $\kappa \not=0$ and $\nabla \times \m + \kappa \m=0$ for some measurable $\m$ field with $|\m|=M$ on a connected domain in $\R^3$.
Then $\m$ is a helix of pitch $2\pi/\kappa$, i.e., there exists $\bs{R} \in SO(3)$ such that 
\[
\m(x)= M ( \bs h_{\bs{R}})(\kappa \x) \quad \text{where} \quad 
\bs h(\x)=(0, \cos x_1, \sin x_1).
\]
\end{lemma}
 
\begin{proof} Upon rescaling one may assume $M=1$ and $\kappa=1$. Taking the divergence it follows that $\nabla \cdot \m= 0$ 
and hence $\Delta \m + \m =0$ in the sense of distributions by taking the curl. So $\m$ is smooth by virtue of standard elliptic regularity theory. In particular, it is enough to prove the claim locally. Denoting the componentwise gradient
by $(\nabla \m)_{jk}:=( \del_j m_k)$ we claim that 
\begin{equation} \label{eq:rank1_prelim}
\mathrm{rank }(\nabla \m) =1 \quad \text{and} \quad (\nabla \m)^2=0. 
\end{equation}
In fact, we observe that for 
arbitrary smooth fields
\begin{equation} \label{eq:div_lap}
\nabla \cdot \left( (\m \cdot \nabla) \m \right) = \tr (\nabla \m)^2 + (\m \cdot \nabla)( \nabla \cdot \m)
\end{equation}
and by using the assumptions 
\[
(\m \cdot \nabla) \m = (\nabla \m)^T \m= \left( (\nabla \m)^T - (\nabla \m) \right) \m = (\nabla \times \m) \times \m=0.
\]
Collecting all these facts we obtain
\begin{equation}\label{eq:kernel_m}
\m \in \ker (\nabla \m) \cap \ker (\nabla \m)^T
\end{equation}
and 
\begin{equation}\label{eq:traces}
\tr (\nabla \m)= \tr (\nabla \m)^2=0.
\end{equation}
Fixing a point $x$ we may assume after rotation $\m(x)=\ein_3$, so that by \eqref{eq:kernel_m} the matrix
$\nabla \m(x)$ is given by a $2 \times 2$ matrix $A$ such that $\tr A=0$ and $\tr A^2=0$ by \eqref{eq:traces}. 
It is easy to see that $A^2=0$ which implies $\det A=0$.

According to \eqref{eq:rank1_prelim} there exist local smooth unit vector fields $\bs X$ and $\bs Y$ and a function $\lambda$ such that 
\[
\nabla \m = \lambda \bs X \otimes \bs Y \quad \text{and} \quad \m = \bs X \times \bs Y.
\]
Now $\nabla \times \m = \lambda \, \bs X \times \bs Y = \lambda\, \m$, hence $\lambda=-1$ and $\nabla \m = - \bs X \otimes \bs Y$.

\medskip

Assuming $\bs X$ is constant so that after rotation $\bs X=\ein_1$, it follows that $\m=\m(x_1)$ and $m_1=0$.
Now the equation implies for the remaining components $m_2'=-m_3$ and $m_3'=m_2$, so that after a
rotation around the $\ein_1$ axis, $\m=\bs h$. 

To show that $\bs X= const.$ one may invoke the spectral theorem. In fact, symmetry of $\nabla \bs X$ follows from the
symmetry of $\nabla^2 \m = (\nabla \otimes \nabla) \m$ since
\[
- \nabla^2\m = \nabla \bs X \otimes \bs Y + \bs X \otimes \nabla \bs Y \quad \text{so that} \quad - \nabla^2\m \cdot \bs Y = \nabla \bs X.
\]
Next we use the identity
\[
 \m = - \Delta \m = (\nabla \cdot \bs X) \bs Y + (\bs X \cdot \nabla) \bs Y, 
 \]
which, after multiplication by $\bs Y$, implies that $\tr( \nabla \bs X) = \nabla \cdot \bs X=0$. Now since $\bs X \in \ker \nabla \m$
\[
 (\nabla \bs X) \m=\nabla (\bs X \cdot \m)= 0
\]
so that $\bs X,\m \in \ker (\nabla \bs X)$, and in turn $\nabla \bs X=0$. \end{proof}

\section{Higher order terms in the bifurcation solution}\label{sec:higherorder}

We have the following analytic expansion for the bifurcation solution
\[
\bs \m_s = \sum_{k=1}^\infty s^k\bs \varphi_k , \quad
\lambda_s = \lambda_0 + \sum_{k=1}^\infty s^k \nu_k 
\]
where the $\bs \varphi_k \in H_\Lambda^2$ satisfy
\begin{equation}\label{eq:orthogonal}
\scp{\bs \varphi_k, \bs \varphi_1} = \fint_{\Omega_\Lambda} \bs \varphi_k(x) \cdot \bs \varphi_1 (x) \dif x =0 \quad\text{for } k \ge 2.
\end{equation}
Inserted into \eqref{eq:main}, we obtain a hierarchy of equations in orders of $s$ that can be solved successively.
Solutions, which are fixed under the symmetry group, can be found by the Fourier method used for solving \eqref{eq:linear}.

The first order equation is the linearized equation
\[
-\lap \bs \varphi_1 + 2 \kappa \grad \times \bs \varphi_1 + \lambda_0 \bs \varphi_1 + \beta (\bs \varphi_1)_3=0
\]
which is certainly fulfilled. By normalizing in terms of $\scp{\abs{\bs \varphi_1}^2}^{1/2}$, the second order equation becomes
\begin{equation}\label{eq:second}
-\lap \bs \varphi_2 + 2 \kappa \grad \times \bs \varphi_2 + \lambda_0 \bs \varphi_2 + \beta (\bs \varphi_2)_3 + \nu_1\frac{\bs \varphi_1}{\scp{\abs{\bs \varphi_1}^2}^{1/2}} =0.
\end{equation}
Multiplying this equation with $\bs \varphi_1$ and integrating over $\Omega_\Lambda$ yields
\[
\scp{L_0\bs \varphi_2, \bs \varphi_1}+ \nu_1 \scp{\abs{\bs \varphi_1}^2}^{1/2}=0
\]
hence $\nu_1=0$ and in turn $L_0 \bs \varphi_2 = 0$ by \eqref{eq:second} which implies $\bs \varphi_2 \equiv 0$ by \eqref{eq:orthogonal}.

The third order equation is
\begin{equation}\label{eq:third}
-\lap \bs \varphi_3 + 2 \kappa \grad \times \bs \varphi_3 + \lambda_0 \bs \varphi_3 + \beta (\bs \varphi_3)_3 +\alpha\frac{\abs{\bs \varphi_1}^2 \bs \varphi_1}{\scp{\abs{\bs \varphi_1}^2}^{3/2}}+ \nu_2 \frac{\bs \varphi_1}{\scp{\abs{\bs \varphi_1}^2}^{1/2}} =0.
\end{equation}
Multiplying this equation with $\bs \varphi_1$ and integrating over $\Omega_\Lambda$ yields
\[
\nu_2 = -\alpha \frac{\scp{\abs{\bs \varphi_1}^4}}{\scp{\abs{\bs \varphi_1}^2}^2}
\]
where $\bs \varphi_3$ is a complicated function in $H_\Lambda^2$ satisfying $\scp{\bs \varphi_1, \bs \varphi_3}=0$.

The forth order equation
\[
-\lap \bs \varphi_4 + 2 \kappa \grad \times \bs \varphi_4 + \lambda_0 \bs \varphi_4 + \beta (\bs \varphi_4)_3 + \nu_3 \frac{\bs \varphi_1}{\scp{\abs{\bs \varphi_1}^2}^{1/2}} =0
\]
is identical to \eqref{eq:second}, so that $\nu_3=0$ and $\bs \varphi_4 \equiv 0$ by the same argument.

Substituting the bifurcation solution into the average energy \eqref{eq:averageEnergy} yields
\begin{align*}
E_{\Lambda}(\m_s,\lambda_s)&= \frac{s^2}{2}\scp{L_0\bs \varphi_1,\bs \varphi_1} + s^4\scp{L_0\bs \varphi_1,\bs \varphi_3}
+\frac{s^4}{\abs{\Omega_\Lambda}}\int_{\Omega_\Lambda} 
\frac{\nu_2}{2} \frac{\abs{\bs \varphi_1}^2}{\scp{\abs{\bs \varphi_1}^2}} + \frac{\alpha\abs{\bs \varphi_1}^4}{4\scp{\abs{\bs \varphi_1}^2}^2} \dif x + O(s^6)\\
 &= \frac{s^4}{4} \bra*{-\alpha\frac{\scp{\abs{\bs \varphi_1}^4}}{\scp{\abs{\bs \varphi_1}^2}^2}}+O(s^6).
\end{align*}
\nocite{*}
\bibliographystyle{abbrv}
\bibliography{references}

\end{document}